\documentclass[10pt,leqno]{amsart}
\usepackage{graphicx}
\usepackage{indentfirst,csquotes,fancyhdr}

\usepackage[utf8]{inputenc}
\usepackage{enumitem,mathtools,amsmath,amsfonts,amssymb,amsthm,relsize,bm}
\usepackage{breakurl,multirow} 
\usepackage{longtable,subcaption}

\usepackage{thmtools,thm-restate}
\usepackage{url}
\usepackage{doi}

\usepackage{dsfont,algorithmic,algorithm}
\usepackage{tikz}
\usepackage{eqparbox}

\newcommand\LONGCOMMENT[1]{%
  /*\;\;\hfill\begin{minipage}[t]{0.9\textwidth}\itshape #1\strut\hfill*/\end{minipage}%
}

\usetikzlibrary{automata, positioning, arrows,patterns,shadows,shapes,datavisualization,calc}
\usetikzlibrary{arrows.meta,patterns.meta,graphs}
\usetikzlibrary{decorations.pathreplacing,decorations.pathmorphing}

\definecolor{myGray}{RGB}{80,80,80}
\definecolor{myRed}{RGB}{250,200,200}
\definecolor{myBlue}{RGB}{220,220,250}

\DeclareRobustCommand{\bigO}{%
  \text{\usefont{OMS}{cmsy}{m}{n}O}%
}

\def\VarSet{\mathbb{X}}

\def\supr{\sqcup}
\def\infn{\sqcap}
\newcommand\latop{\mathbin{\raisebox{-0.33ex}{\ooalign{\hss\raisebox{0.5ex}{$\supr$}\hss\cr$\infn$}}}}
\newcommand\suprcart{\,\mathbin{\mathrlap{\scalebox{1.75}{\raisebox{-0.3ex}{\hspace*{-0.3ex}$\supr$}}}\scalebox{0.65}{$\tcartesian$}}\,}
\newcommand\infncart{\,\mathbin{\mathrlap{\scalebox{1.75}{\raisebox{-0.3ex}{\hspace*{-0.3ex}$\infn$}}}\scalebox{0.65}{$\tcartesian$}}\,}

\def\oxbr#1{[\hspace*{-0.35ex}[{#1}]\hspace*{-0.35ex}]}

\def\Lang{\mathcal{L}}
\def\Sol{\mathit{\mathbf{Sol}}}

\def\lenmorph{\sigma_{\forall\rightarrow \nullsymb}}

\def\JS{\textsc{JavaScript}}
\def\Eq{\mathit{E}}
\def\Pat{\mathit{P}}
\def\logimpl{\mathrel{\Rightarrow}}
\def\logand{\mathrel{\&}}
\def\logor{\mathrel{\vee}}

\def\nullsymb{\delta_{\varnothing}}

\def\fancyp{\bm{+}}
\def\fancym{\bm{-}}
\def\fancypm{\bm{\pm}}

\def\WL1{\mathcal{WL}_1}

\def\quote{\hspace{-0.2ex}\raisebox{-0.4ex}{\rotatebox[origin=c]{13}{${}^{{\tiny{\prime}}}$}}}

\def\empt{\varepsilon}

\def\infopat#1{\mathit{E}^{#1}}
\def\absdom{\mathcal{L}^{\#}}

\def\absval{\nu^\alpha}

\def\StringObjectDom{\mathcal{SOB}}

\newcommand\SOb[3]{\biggl\langle {#1}\logand \text{\textbf{len}}_Z\in[{#2},\,{#3})\biggr\rangle}
\def\objlen#1{\text{\textbf{len}}(#1)}
\def\objval#1{\text{\textbf{val}}(#1)}
\newcommand\objprop[2]{\text{\textbf{prop}}_{#2}(#1)}
\newcommand\weql{{=}}
\newcommand\wequat[2]{#1 = #2}
\newcommand\nequat[2]{#1\neq #2}
\newcommand\oconcat[2]{#1.\mathtt{concat}(#2)}

\newcommand\osubstr[2]{#1.\mathtt{substring}(#2)}
\newcommand\oreplace[3]{#1.\mathtt{replace}(#2,\,#3)}
\newcommand\oindex[2]{#1.\mathtt{indexOf}(#2)}
\newcommand\ochar[2]{#1.\mathtt{charAt}(#2)}
\newcommand\concr[1]{\gamma\bigl(#1\bigr)}
\newcommand\abstr[1]{\alpha\bigl(#1\bigr)}

\def\CAlph{\Sigma}
\newtheorem{theorem}{Theorem}[section]
\newtheorem{lemma}{Proposition}[section]
\newtheorem{definition}{Definition}[section]
\newtheorem{example}{Example}[section]

\tikzset{
    ncbar angle/.initial=90,
    ncbar/.style={
        to path=(\tikztostart)
        -- ($(\tikztostart)!#1!\pgfkeysvalueof{/tikz/ncbar angle}:(\tikztotarget)$)
        -- ($(\tikztotarget)!($(\tikztostart)!#1!\pgfkeysvalueof{/tikz/ncbar angle}:(\tikztotarget)$)!\pgfkeysvalueof{/tikz/ncbar angle}:(\tikztostart)$)
        -- (\tikztotarget)
    },
    ncbar/.default=0.5cm,
}

\tikzset{square left brace/.style={ncbar=0.5cm}}
\tikzset{square right brace/.style={ncbar=-0.5cm}}

\tikzset{round left paren/.style={ncbar=0.5cm,out=120,in=-120}}
\tikzset{round right paren/.style={ncbar=0.5cm,out=60,in=-60}}

\newcommand{\eqinterval}[2]{
\begin{tikzpicture}[line width=0.8pt, color=white!10!gray, scale=#1,every node/.style={inner sep=-8pt}]

\draw[]  (0,0.5)   edge [draw=gray!90!white,decorate,decoration={zigzag,segment length=(#1+1)/2*0.4em, amplitude=(#1+1)/2*0.3mm}] (0.75,0.5) (0.75,0.5) -- (0.75,0) (0.75,0) edge[draw=black,line width=1.1pt] (0,0) (0,0) -- (0,0.5);
\draw[]  (0,0) edge [draw=black,line width=1.1pt] (0,0.15);
\draw[]  (0.75,0) edge [draw=black,line width=1.1pt] (0.75,0.15);

\node[text=black!30!gray] at (0.375,0.25) {#2};
\end{tikzpicture}
}

\tikzset{
  pics/eqnode/.style args={#1,#2,#3,#4,#5}{
     code={
    \draw[]  (0,#4*#5*0.5) edge [draw=gray!90!white,decorate,decoration={zigzag,segment length=(#4+1)/2*0.4em, amplitude=(#4+1)/2*0.3mm}] (#4*0.75,#4*#5*0.5) (#4*0.75,#4*#5*0.5) -- (#4*0.75,0) (#4*0.75,0) edge [draw=black,line width=1.1pt] (0,0) (0,0) -- (0,#4*#5*0.5);
    \draw[]  (0,0) edge [draw=black,line width=1.1pt] (0,#4*#5*0.15);
    \draw[]  (#4*0.75,0) edge [draw=black,line width=1.1pt] (#4*0.75,#4*#5*0.15);
     \node[] (#1) at (#4*0.375,#4*#5*0.25) {#2};
     }
  }
}

\newcommand{\neqinterval}[2]{
\begin{tikzpicture}[line width=0.8pt, color=white!10!gray, scale=#1,every node/.style={inner sep=-8pt}]

\draw[]  (0,0.4)   edge[out=30,in=150,looseness=1,draw=black,line width=1.1pt] (0.75,0.4) (0.75,0.4) -- (0.75,0) (0.75,0) edge[draw=gray!90!white,decorate,decoration={zigzag,segment length=(#1+1)/2*0.4em, amplitude=(#1+1)/2*0.3mm}] (0,0) (0,0) -- (0,0.4);
\node[text=black!30!gray] at (0.375,0.225) {#2};
\end{tikzpicture}
}

\newcommand{\Cartesian}[1]{\mathbin{\begin{tikzpicture}[x=#1ex,y=#1ex,line width=(#1)/4.0*1ex] \draw (-0.7,-0.9) -- (0.7,0.5) (-0.7,0.5) -- (0.7,-0.9); \draw (0,-0.2) circle (1.41);\end{tikzpicture}}}

\newcommand{\tcartesian}{\raisebox{-0.2ex}{\,$\Cartesian{0.6}$\,}}

\newcommand{\Reduced}[1]{\mathbin{\begin{tikzpicture}[x=#1ex,y=#1ex,line width=(#1)/4.0*1ex] \draw (-0.7,-0.7) -- (0,0.7) (0.7,-0.7) -- (0,0.7); \draw (0,0) circle (1.41);\end{tikzpicture}}}

\newcommand{\treduced}{\raisebox{-0.2ex}{\,$\Reduced{0.6}$\,}}

\newcommand{\Point}{\mathcal{C}}
\newcommand{\Prefix}{\mathcal{P}}
\newcommand{\Suffix}{\mathcal{S}}
\newcommand{\FactorCode}{\mathcal{I}}
\newcommand{\AntiDictionary}{\mathcal{F}}
\newcommand{\Bool}{\mathcal{B}}
\newcommand{\StringProperty}{\mathcal{SP}}

\newcommand{\objfullreprN}[3]{\biggeral \objval{#1}\treduced \objlen{#1}\treduced \raisebox{-0.8ex}{$\Reduced{1.2}$}_{i=#2}^{#3}\objprop{#1}{i}\biggerar}

\newcommand{\syminterval}[2]{
\begin{tikzpicture}[line width=0.8pt, color=white!60!gray, scale=#1,every node/.style={inner sep=-8pt}]

\draw[]  (0,0.45)   edge[out=30,in=150,looseness=1,draw=black,line width=1.1pt] (0.75,0.45) (0.75,0.45) -- (0.75,0) (0.75,0) edge[draw=black,line width=1.1pt] (0,0) (0,0) -- (0,0.45);
\node[text=black!30!gray] at (0.375,0.25) {#2};
\draw[]  (0,0) edge [draw=black,line width=1.1pt] (0,0.15);
\draw[]  (0.75,0) edge [draw=black,line width=1.1pt] (0.75,0.15);
\end{tikzpicture}
}

\makeatletter
\newcommand{\biggerl}{\bBigg@{1.15}({}\hspace*{-0.25ex}} 
\newcommand{\biggerr}{\hspace*{-0.25ex}\bBigg@{1.15})} 
\newcommand{\biggeral}{\bBigg@{1.25}\langle{}\hspace*{-0.1ex}} 
\newcommand{\biggerar}{\hspace*{-0.1ex}\bBigg@{1.25}\rangle} 
\makeatother

\begin{document}
\title{Abstract String Domain Defined \\with Word Equations as a Reduced Product
}
\date{\today}
\author{Antonina Nepeivoda \and Ilya Afanasyev}

\maketitle
\begin{abstract}
We introduce a string-interval abstract domain, where string intervals are characterized by systems of word equations (encoding lower bounds on string values) and word disequalities (encoding upper bounds). Building upon the lattice structure of string intervals, we define an abstract string object as a reduced product on a string property semilattice, determined by length-non-increasing morphisms. We consider several reduction strategies for abstract string objects and show that upon these strategies the string object domain forms a lattice. We define basic abstract string operations on the domain, aiming to minimize computational overheads on the reduction, and show how the domain can be used to analyse properties of JavaScript string manipulating programs.
\end{abstract}


\section{Introduction}

Finding program invariants is one of the main goals of static program analysis, and the set of invariants discovered by an analyser must be expressed in an appropriate language. Therefore, the choice of the language used to describe these invariants is a crucial aspect of any program analysis, as it determines the methods employed in the analysis itself.

Languages capturing numeric invariants naturally involve equations. For example, the invariant \emph{``the value of variable 
$Z$ is never less than $n$''} can be expressed by the formula
$\exists X\geq 0\bigl(Z=X+n\bigr)$
meaning that the desired values of $Z$ belong to the 
$Z$-projection of the solution set of equation $\wequat{Z}{X+n}$ with the constraint $X\geq 0$. In the context of abstract interpretation, well-known abstract domains --- such as intervals~\cite{Intervals}, octagons, and polyhedra~\cite{ComposingBoxes} --- are all based on systems of linear equations. Once the language of equations is chosen, the full power of linear algebra can be leveraged in program analysis. The domain may be relational (capturing relationships between values) or non-relational (capturing properties of individual values), with the latter interpretable as projections of solutions to multi-variable equations.

For non-numeric data, the choice of a language remains an active area of research. For instance, string constraints can be expressed using grammars, automata, or Boolean formulas involving predefined string predicates~\cite{Regular}. Word equations are also widely used to capture string invariants~\cite{Rummer,Rummer2}. But in abstract interpretation, the idea of constructing word-equations-based domains is relatively new and largely unexplored~\cite{N}.

Given a set of variables $\VarSet$ and a set of constant letters $\CAlph$, a word equation is an equality $\wequat{\Phi}{\Psi}$, where $\Phi,\Psi\in (\VarSet\cup\CAlph)^*$. An expressive power of a word equation language depends on the constraints imposed on the equation’s form. For example, the work~\cite{N} introduces a simple abstract domain based on one-variable word equations, where $\Phi\Psi$ includes occurrences of a single variable. But even for expressing non-relational program properties, equations with multiple variables are meaningful. Analogous to numerical non-relational domains, the corresponding formal languages in this case are projections of equation solutions. For example, linear word equations, i.e. equations with each variable occurring in $\Phi\Psi$ at most once, can capture properties of dense languages (defined below). While solutions of these equations are regular, their representation in terms of regexes or NFA can be exponentially larger than in terms of the equation systems (see~\cite{Birget} and Lemma~\ref{Lemma::NFA} in Appendix).

A formal language $\Lang$ is called dense, if for any $\omega\in\Sigma^+$, the intersection $\Lang\cap \Sigma^* \omega \Sigma^*$ is non-empty~\cite{Day}. Such languages are practically interesting because they can capture some of invariants on results of string operations over unknown values. If a language is not dense, it is called thin.

There are many subtle problems of the string domain, as compared to the numeric domain, having impact on expressible power of word equations. For example, introduction of a linear order on strings can result in a non-monotonicity of usual string operations, e.g. given a lexicographic order with alphabet ordering $a\prec b$, $aab\prec b$, while $aab$ includes $b$ as a substring. Introduction of a monotone wrt the subword relation string order usually involves \emph{length comparison}, and the length property is inexpressible in terms of mere word equations~\cite{Day}. 

Hence, we are interested in finding such an abstract domain that is  based on equations whose solution sets are able to express properties of dense languages, and can capture the string length property in a natural way.

The main contributions of the paper are following. 

\begin{itemize}
\item We introduce a notion of a string-interval abstract domain, being a generalization of a numeric interval abstract domain. The string intervals are described via systems of word equations and word disequalities to be satisfied by the strings in the concretisation set of the abstract value. We show that the equation-based string interval domain satisfies lattice axioms, and describe the reduction constructing an unique representation of the string interval in a given alphabet.
\item We introduce a notion of an abstract string object, being described in terms of string intervals and morphisms, generalizing the notion of a $\JS$ string object having value and length properties.
\item We discuss reduction strategies for abstract string objects, preserving their consistency as lattice elements, and describe abstract semantics of basic string operations. We show how the abstract string objects can be used for analysing invariants of string manipulating programs.
\end{itemize}

After giving preliminary definitions in Sect.~\ref{Sect:prelim}, we introduce string intervals (Sect.~\ref{Sect:interval}) and assemble them in objects (Sect.~\ref{Sect:object}). Reduction is discussed in Sect.~\ref{Sect:reduction}, Sect.~\ref{Sect:operations} is devoted to string operation abstraction, and contains program examples analysed. Sect.~\ref{Sect:last} discusses related works and concludes the paper. 

\section{Preliminaries}\label{Sect:prelim}

Let $\Sigma$ denote an alphabet. Small Latin letters $a$, $b$, $c$, $d$, and a letter parameter $\delta$ are considered to be elements of $\Sigma$ ($\delta$ can be also considered as an unknown one-letter string). Small Greek letters $\tau$, $\omega$, $\upsilon$ stand for words (elements of $\Sigma^*$), Greek letters $\sigma$, $\eta$, $\theta$ are reserved for morphisms, small Greeks $\alpha$ and $\gamma$ are reserved for special operations of abstraction and concretisation, and $\nu$ is used for denoting abstract values. 

Capital Latin letters $X$, $Y$, $Z$ stand for elements of the variable alphabet $\VarSet$. The notation $\tau^n$ stands for $n$-concatenation of $\tau$ with itself, i.e. $\underbrace{\tau\tau\dots\tau}_n$. The empty word is denoted by $\empt$. Given a 
word $\tau$ and $t\in \Sigma \cup \VarSet$, $|\tau|$ is the length of $\tau$, and $|\tau|_t$ stands for the number of occurrences of $t$ in $\tau$. We denote an application of a morphism $\sigma$ to $\rho\in (\Sigma\cup \VarSet)^*$ either by $\sigma(\rho)$, or by $\rho[t\mapsto \nu]$ if $\sigma$ is an identity on all elements of $\Sigma\cup\VarSet\setminus\{t\}$.

\begin{definition}
Given a letter alphabet $\Sigma$ and a variable alphabet $\VarSet$, \emph{a word equation} is an equation $\wequat{\mathcal{U}}{ \mathcal{V}}$, where $\mathcal{U},\mathcal{V}\in (\Sigma\cup \VarSet)^*$. The equation is linear iff for every $X\in\VarSet$, $|\mathcal{U}\mathcal{V}|_X\leq 1$.

\emph{A solution} to an equation $\wequat{\mathcal{U}}{ \mathcal{V}}$ is a morphism $\sigma:(\VarSet\cup\Sigma)^*\rightarrow \Sigma^*$ preserving elements of $\Sigma$, s.t. $\sigma(\mathcal{U})=\sigma(\mathcal{V})$~\cite{Plandowski,Makanin}.

We assume that the variable set of equation $\wequat{\mathcal{U}}{\mathcal{V}}$ is lexicographically ordered. The solution set of $\wequat{\mathcal{U}}{ \mathcal{V}}$ is the set of tuples of variable images determined by all solutions of $\wequat{\mathcal{U}}{\mathcal{V}}$ Given a variable set $\mathcal{Q}$, the solution set of $\wequat{\mathcal{U}}{ \mathcal{V}}$ wrt $\mathcal{Q}$, denoted $\Sol_{\mathcal{Q}}$, is the projection of the whole solution set of $\wequat{\mathcal{U}}{ \mathcal{V}}$ onto the coordinates corresponding to the elements of $\mathcal{Q}$. 
\end{definition}

Henceforth, given an equation or an equation system depending on a set of variables including $Z$, we are interested in its $Z$-solution set\footnote{For the sake of brevity, we omit the set notation in $\Sol_{\mathcal{Q}}$ if $|\mathcal{Q}|=1$.} $\Sol_Z$, considering its as a default language defined by the equation. E.g. $\wequat{ZX}{a}$ defines $\Sol_Z$ equal to $\{\empt,a\}$; and $\Sol_Z(\wequat{ZXaY}{XaYZ})=\{\empt\}\cup\Sigma^* a \Sigma^*$. In the similar way, given a word disequality $\nequat{\Phi}{\Psi}$, where $X_1, X_2, ..., X_n, Z\in \VarSet$ are all variables occurring in $\Phi\Psi$, its $Z$-solution set $\Sol_Z$ is a set of all $\omega$ s.t. $\Phi[Z\mapsto\omega]\neq\Psi[Z\mapsto\omega]$ holds for any values of $X_1,\dots, X_n$. For example, $\Sol_Z(\nequat{ZZ}{X_1 a X_2})$ is the set of all words in the alphabet $\Sigma\setminus\{a\}$. Indeed, if $\sigma(Z)=\omega_1 a\omega_2$, then for $\sigma(X_1)=\omega_1$, $\sigma(X_2)=\omega_2 \omega_1 a \omega_2$ the disequality does not hold. 

\subsection{Lattices and Reduced Products}

\begin{definition}
A triple $\bigl\langle \Lang_{\Delta},\supr,\infn\bigr\rangle$, where $\Lang_{\Delta}$ is a set, $\supr$ and $\infn$ are binary operations over $\Lang_{\Delta}$ (also called join and meet respectively), is said to be a lattice if it satisfies the following axioms~\cite{Cousot} for all $x,y,z\in\Lang_{\Delta}$:

\begin{itemize}
\item $\bigl(x\supr(x\infn y)= x\bigr)\logand \bigl(x\infn (x\supr y) = x\bigr)$;

\item $\bigl(x\supr y = y \supr x\bigr)\logand \bigl(x\infn y = y\infn x\bigr)$;

\item $\bigl(x\supr(y\supr z) = (x\supr y)\supr z\bigr)\logand \bigl(x\infn (y\infn z) = (x\infn y)\infn z\bigr)$.
\end{itemize}

A partial order on $\mathcal{L}_{\Delta}$ induced by the lattice above is defined as follows: $x\preceq y \equiv (x\supr y=y)$.

\end{definition}

Let $S$ and $\Lang_\Delta$ be sets of concrete and 
abstract data values, s.t. $\bigl\langle \Lang_\Delta,\supr,\infn\bigr\rangle$ is a lattice. Consider the following two functions:
\begin{itemize}
\item {Abstraction} $\alpha: 2^S\rightarrow {\Lang_\Delta}$;
\item {Concretisation} $\gamma: \Lang_\Delta\rightarrow 2^{S}$. 
\end{itemize}

The functions define a Galois connection iff for all $A\in 2^S$ $A\subseteq \concr{\alpha(A)}$ and for all $\rho \in \Lang_{\Delta}$ $\rho \preceq \abstr{\gamma(\rho)}$.

Henceforth we say that $\absdom$ is an abstract domain of data $S$ assuming that $\absdom$ is a lattice and it forms a Galois connection with $S$.

A lattice being an abstract domain $E$ is said to be atomistic, if for any element $\omega\in S$ $\concr{\abstr{\{\omega\}}}=\{\omega\}$, i.e. any singleton defines a corresponding unique abstract value, which is called an atom~\cite{LatticeAutomata}. Atomistic lattices are of special interest of abstract interpreters, since such an abstract interpreter is able to perform constant propagation ``for free''.

Given two abstract domains $\absdom_1=\langle\mathcal{L}_1,\supr_1,\infn_1\rangle$, $\absdom_2=\langle\mathcal{L}_2,\supr_2,\infn_2\rangle$ abstracting the same concrete set $S$, their Cartesian product, defined as 

\vspace*{-2.7ex}
\begin{multline*}\absdom_1\tcartesian\absdom_2=\biggeral (\mathcal{L}_1,\mathcal{L}_2),\lambda (\nu_1,\nu_2),(\nu'_1,\nu'_2).(\nu_1\supr_1 \nu'_1,\nu_2\supr_2 \nu'_2),\\\lambda (\nu_1,\nu_2),(\nu'_1,\nu'_2).(\nu_1\infn_1 \nu'_1,\nu_2\infn_2 \nu'_2)\biggerar,
\end{multline*} 

is also an abstract domain. Here we use the usual $\lambda$ notation. This definition is naturally extended on any number of elements of the product.

If the abstract domains are not completely independent (i.e. the domains capture related properties of concrete data), then certain subsets of the Cartesian product can describe the same concretisation set of $S$ elements. For the purpose of tracking such subsets, the notion of reduction and reduced product is introduced in papers~\cite{CousotReduced,ComposingDomains}.

\begin{definition}
Given an abstract domain $\absdom$, function $\rho: \absdom\rightarrow \absdom$ is said to be \emph{a reduction} iff $\forall \nu\in\absdom\bigl(\concr{\nu}=\concr{\rho(\nu)}\logand \rho(\nu)\preceq \nu\bigr)$. 

\end{definition}

In this paper we require that $\rho\circ\rho=\rho$. In general a reduction may be non-idempotent and even non-stabilizing for any number of steps~\cite{CousotReduced}.

\begin{example}\label{Example::IntIntervalReduced}

Let us consider a product of abstract numeric integer domains tracking minimal and maximal possible values of unknowns in $\mathbb{Z}^+$. Elements of the both domains can be represented with extended integers in $\mathbb{Z}^+\cup\,\{+\infty\}$, or with $\bot$. Let $\mathcal{LB}$ be the abstract domain representing closed lower bound of value sets, and $\mathcal{UB}$ be the abstract domain representing open upper bound. The concretisation set of the value $\langle \nu_1, \nu_2\rangle\in \mathcal{LB}\tcartesian\mathcal{UB}$ s.t. $\nu_1 \geq \nu_2$ is $\varnothing$. Hence, such a value can be reduced to $\langle\bot, \bot \rangle$, as well as values in which some component of the pair is equal to $\bot$. The following function serves as a reduction for $\mathcal{LB}\treduced\mathcal{UB}$:

$$\rho(\langle\nu_1, \nu_2\rangle)=\begin{cases}\langle \bot,\bot\rangle,\text{ if }\nu_1=\bot \vee \nu_2 =\bot\vee \nu_1 \geq \nu_2,\\
\langle\nu_1, \nu_2\rangle, \text{ otherwise}.\end{cases}$$

\end{example}

Equivalence classes wrt $\rho$, i.e. sets $S_i$ s.t. $\forall \nu_1\in S_i, \nu_2\in S_i\bigl(\rho(\nu_1)=\rho(\nu_2)\bigr)$, determine a quotient set of $\absdom$, $\absdom_{/\equiv_{\rho}}$. Given a reduction function on Cartesian product $\absdom_1\tcartesian \absdom_2\dots \tcartesian\absdom_k$, we define \emph{a reduced product} $\absdom_1\treduced \absdom_2\dots \treduced \absdom_k = \langle \mathcal{L}_r,\supr_r,\infn_r\rangle$, shortcut as $\raisebox{-0.75ex}{$\Reduced{1.2}$}_{i=1}^n \absdom_i$, as follows.

\begin{itemize}
\item $\mathcal{L}_r=(\absdom_1\tcartesian \absdom_2\dots \tcartesian\absdom_k)_{/\equiv_{\rho}}$;
\item $\supr_r = \lambda (\nu_1,\dots,\nu_k),(\nu'_1,\dots,\nu'_k).\rho\bigl((\nu_1\supr_1 \nu'_1,\dots,\nu_k\supr_k \nu'_k)\bigr)$;
\item $\infn_r = \lambda (\nu_1,\dots,\nu_k),(\nu'_1,\dots,\nu'_k).\rho\bigl((\nu_1\infn_1 \nu'_1,\dots,\nu_k\infn_k \nu'_k)\bigr)$.
\end{itemize} 

Note that, in contrast to Cartesian product (which can be considered as a reduced product with $\rho=id$), arbitrary reduction of a lattice is not necessarily a lattice. Given pointwise join and meet operations, $\suprcart$ and $\infncart$ on the Cartesian product, in order to guarantee that the reduced product forms a lattice, the reduction function $\rho$ is to satisfy the following properties, assuming $x$, $y$, $z$ are already given in the normal form (i.e. $\rho(x)=x$, $\rho(y)=y$, $\rho(z)=z$):

\begin{itemize}
\item $\biggerl\rho\bigl(x\suprcart\rho(x\infncart y)\bigl)= x\biggerr\logand \biggerl\rho\bigl(x\infncart \rho(x\suprcart y)\bigr) = x\biggerr$;

\item $\biggerl\rho\bigl(x\suprcart\rho(y\suprcart z)\bigr) = \rho\bigl(\rho(x\suprcart y)\suprcart z\bigr)\biggerr\logand \biggerl\rho\bigl(x\infncart\rho(y\infncart z)\bigr) = \rho\bigl(\rho(x\infncart y)\infncart z\bigr)\biggerr$.
\end{itemize}

\section{String Properties}\label{Sect:interval}

Let $s\in \Sigma$, and $m,n\in\mathbb{N}$ be given. For a unary Peano number $Z$ represented in a string notation, word equation $\wequat{Z}{X s^n Y}$ determines numeric disequality $Z\geq n$, i.e. the lower bound of the interval containing possible values of $Z$. Any word disequality $\forall X,Y(\nequat{Z}{X s^m Y})$ determines numeric disequality $Z< m$. Thus, we can define any natural interval of values by means of word equations and disequalities.

In the similar way, given alphabet $\Sigma$, we can say that an equation $\wequat{Z}{X\omega Y}$, $\omega\in\Sigma^+$, determines a lower bound on $Z$ value, while any disequality $\nequat{Z}{X\omega Y}$ determines
an upper bound. A set of such disequalities (\emph{anti-dictionary}~\cite{Crochemore}) determines subword-closed set of words, and a set of linear word equations of the form above determines an ideal in free monoid over $\Sigma$ defined with the corresponding \emph{factor code}~\cite{Crochemore}. Indeed, given an equation system $\begin{cases}Z=X_1 \omega_1 Y_1 \\\dots \\ Z = X_n\omega_n Y_n\end{cases}$, if some $\omega_i$ is a factor of $\omega_j$, i.e. $\omega_j = \omega_{j,1}\omega_i\omega_{j,2}$, then the equation $Z=X_i\omega_i Y_i$ is redundant in the system, because $Z = X_j \omega_j Y_j$ implies that whenever given $X_i = X_j \omega_{j,1}$, $Y_i = \omega_{j,2} Y_j$. Hence, such systems can be represented by the factor codes of the constant strings in their right-hand sides.

Any finite subword-closed set of words can be described via finite anti-dictionary~\cite{Crochemore,Crochemore2}. In order to define a singleton $\{\omega\}$, both the equation $\wequat{Z}{X\omega Y}$ and the set of corresponding disequalities are required. In general, the disequalities set determining forbidden factors of a word $\omega$ has the size $\bigO\bigl(|\omega|\cdot |\Sigma|\bigr)$. If we are interested in building an atomistic lattice, this representation seems too involved. Hence, we chose to use also equations $\wequat{Z}{\omega}$, and equations of the forms $\wequat{Z}{\omega Y}$ and $\wequat{Z}{X\omega}$. 

It is known that when $|\Sigma|\geq 2$, a union of languages with finite anti-dictionaries may have an infinite anti-dictionary~\cite{Crochemore}. A simple example is $\Sol_Z(\nequat{Z}{X ab Y})\cup \Sol_Z(\nequat{Z}{X ba Y})$ in $\Sigma=\{a,b\}$. This union has the anti-dictionary described by regex $\bigl(ab^+a\;\vert\; ba^+b\bigr)$, which defines an infinite language. Hence, in this paper we consider only string upper bounds in a unary alphabet. The summary on the types of the basic string intervals are given in Fig.~\ref{Fig::basicint}. There $E(\upsilon)$ and $F(\upsilon)$ denote a single equation and a single disequality, respectively. An example illustrating the idea of the upper and lower bounds is given in Fig.~\ref{Fig:strinterval}. 

\begin{figure}[htb]\centering
\begin{tabular}{lllll}
\;\;Icon\quad\quad &\multicolumn{1}{c}{Name} && \multicolumn{1}{c}{Basic components} & \multicolumn{1}{c}{Semantics}\\\hline
\raisebox{-1.5ex}{\eqinterval{1}{}} & \begin{tabular}{l}string interval with \\ trivial upper bound\end{tabular} & \qquad & $E(\upsilon)=\begin{cases}\wequat{Z}{X_i \upsilon Y_i}\\\wequat{Z}{\upsilon Y_j}\\\wequat{Z}{X_k \upsilon} \\\wequat{Z}{\upsilon}\end{cases}$ & $\{\omega\mid \omega\in\bigcap E(\upsilon_i)\}$\\

\raisebox{-1.5ex}{\neqinterval{1}{}}&\begin{tabular}{l}string interval with \\ trivial lower bound\end{tabular}& & $F(\upsilon)=\nequat{Z}{X_i \upsilon Y_i}$ & $\{\omega\mid \omega\in\bigcap F(\upsilon_i)\}$\\
\raisebox{-1.5ex}{\syminterval{1}{}}
& \begin{tabular}{l}string interval with \\ non-trivial bounds\end{tabular}& & \begin{tabular}{l}equations $E(\upsilon)$ \\ and disequalities $F(\upsilon)$\end{tabular} & \begin{tabular}{l}$\{\omega\mid \omega\in\bigcap E(\upsilon_i)$\quad \;\\\multicolumn{1}{r}{$\logand \omega\in \bigcap F(\upsilon_j)\}$}\end{tabular}
\end{tabular}
\caption{Types of basic string intervals}
\label{Fig::basicint}
\end{figure}

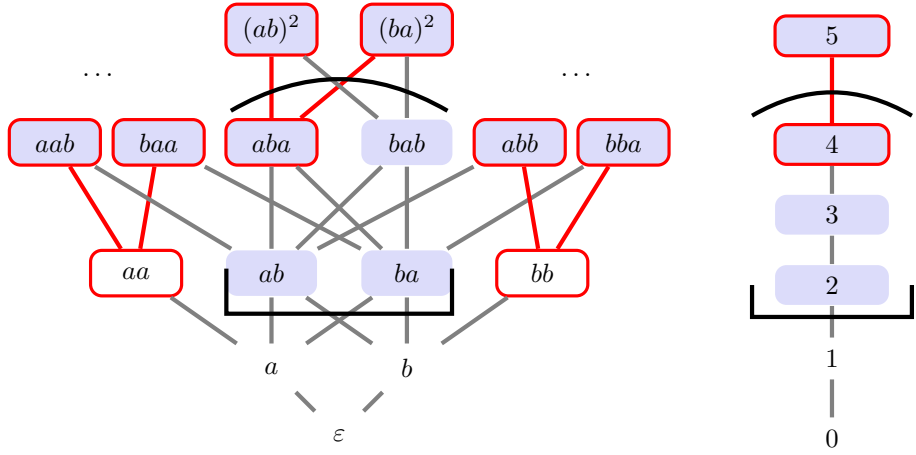
\begin{figure}[htb]

\begin{tabular}{cc}

\begin{tikzpicture}[scale=1.2, every rectangle node/.style={very thick,inner sep=1ex,rounded corners=1ex, minimum width=8ex,minimum height=4ex}]
\draw (0.5, -1) node[rectangle] (0) {$\varepsilon$};
 
\draw (-0.25, -0.25) node[rectangle] (a) {$a$};
\draw (1.25, -0.25) node[rectangle] (b) {$b$};

\draw (-0.25, 0.8) node[rectangle,fill=myBlue] (ab) {$ab$};
\draw (1.25, 0.8) node[rectangle,fill=myBlue] (ba) {$ba$};
\draw (-1.75, 0.8) node[rectangle, draw=red] (aa) {$aa$};
\draw (2.75, 0.8) node[rectangle,  draw=red] (bb) {$bb$};

\draw (-0.25, 2.25) node[rectangle,fill=myBlue, draw=red] (aba) {$aba$};
\draw (1.25, 2.25) node[rectangle,fill=myBlue] (bab) {$bab$};
\draw (-1.5, 2.25) node[rectangle,fill=myBlue, draw=red] (baa) {$baa$};
\draw (-2.65, 2.25) node[rectangle,fill=myBlue, draw=red] (aab) {$aab$};
\draw (2.5, 2.25) node[rectangle,fill=myBlue, draw=red] (abb) {$abb$};
\draw (3.65, 2.25) node[rectangle,fill=myBlue, draw=red] (bba) {$bba$};

\draw (-2.15, 3) node[ellipse, draw=white] (ld) {$\dots$};
\draw (3.15, 3) node[ellipse, draw=white] (rd) {$\dots$};

\draw (-0.25, 3.5) node[rectangle,fill=myBlue, draw=red] (abab) {$(ab)^2$};
\draw (1.25, 3.5) node[rectangle,fill=myBlue, draw=red] (baba) {$(ba)^2$};

\draw[-,ultra thick,gray] (0) -> (a);
\draw[-,ultra thick,gray] (0) -> (b);

\draw[-,ultra thick,gray] (a) -> (aa);
\draw[-,ultra thick,gray] (a) -> (ab);
\draw[-,ultra thick,gray] (a) -> (ba);

\draw[-,ultra thick,gray] (b) -> (bb);
\draw[-,ultra thick,gray] (b) -> (ab);
\draw[-,ultra thick,gray] (b) -> (ba);

\draw[-,ultra thick,red] (aa) -> (aab);
\draw[-,ultra thick,red] (aa) -> (baa);

\draw[-,ultra thick,gray] (ab) -> (aba);
\draw[-,ultra thick,gray] (ab) -> (bab);
\draw[-,ultra thick,gray] (ab) -> (aab);
\draw[-,ultra thick,gray] (ab) -> (abb);

\draw[-,ultra thick,gray] (ba) -> (aba);
\draw[-,ultra thick,gray] (ba) -> (bab);
\draw[-,ultra thick,gray] (ba) -> (bba);
\draw[-,ultra thick,gray] (ba) -> (baa);

\draw[-,ultra thick,red] (bb) -> (abb);
\draw[-,ultra thick,red] (bb) -> (bba);

\draw[-,ultra thick,gray,red] (aba) -> (abab);
\draw[-,ultra thick,gray,red] (aba) -> (baba);
\draw[-,ultra thick,gray] (bab) -> (abab);
\draw[-,ultra thick,gray] (bab) -> (baba);

\draw [ultra thick] (-0.75,0.85) to [square right brace ] (1.75,0.85);
\draw [ultra thick,rotate=270] (-2.6,-0.7) to [round left paren ] (-2.6,1.7);
\end{tikzpicture}
&\qquad
\begin{tikzpicture}[scale=0.85]
\node[rectangle,minimum width=10ex,rounded corners=1ex,thick,fill=white,inner sep=1ex] (0) at (0,-2.5) {$0$};

\node[rectangle,minimum width=10ex,rounded corners=1ex,thick,fill=white,inner sep=1ex] (1) at (0,-1.25) {$1$};
\node[rectangle,minimum width=10ex,rounded corners=1ex,thick,fill=myBlue,inner sep=1ex] (2) at (0,-0.1) {$2$};
\node[rectangle,minimum width=10ex,rounded corners=1ex,thick,fill=myBlue,inner sep=1ex] (3) at (0,1) {$3$};
\node[rectangle,minimum width=10ex,rounded corners=1ex,draw=red,very thick,fill=myBlue,inner sep=1ex] (4) at (0,2.1) {$4$};
\node[rectangle,minimum width=10ex,rounded corners=1ex,draw=red,very thick,fill=myBlue,inner sep=1ex] (5) at (0,3.8) {$5$};

\draw[-,ultra thick,gray] (0) -> (1);
\draw[-,ultra thick,gray] (1) -> (2);
\draw[-,ultra thick,gray][-] (2) -> (3);
\draw[-,ultra thick,gray][-] (3) -> (4);
\draw[-,ultra thick,red][-] (4) -> (5);
\draw [ultra thick] (-1.25,-0.1) to [square right brace ] (1.25,-0.1);
\draw [ultra thick,rotate=270] (-2.55,-1.25) to [round left paren ] (-2.55,1.25);
\end{tikzpicture}

\end{tabular}
\caption{\small String interval wrt the subword ordering represented by the bounds $\omega\in\Sol_Z\{Z = X_1 a Y_1, Z = X_2 b Y_2\}$ and $\omega\in\Sol_Z\{Z\neq X_1 aa Y_1, Z\neq X_2 bb Y_2, Z\neq X_3 aba Y_3\}$ in $\Sigma=\{a,b\}$, and its length interval $[3,4)$. Nodes satifying the lower bound are filled with blue; nodes violating the upper bound condition are circled in red.}\label{Fig:strinterval}
\end{figure}

We treat an equation-based string interval bound as a reduced product in $\Point\treduced(\Prefix\tcartesian\Suffix)\treduced\FactorCode$, where $\Point$ is a lattice capturing constant equations, $\Prefix$ and $\Suffix$ are lattices capturing the prefix and the suffix properties, respectively, and $\FactorCode$ is a lattice capturing the factor code. The reduction performs at least  the following steps.
\begin{itemize}
\item if any element of $\Point$, $\Suffix$, $\Prefix$, or $\FactorCode$ is $\bot$, then reduce to $\bot$ (i.e. to the product of $\bot$ components).
\item if $\nu_{\Point}\neq \top$, where $\nu_{\Point}\in\Point$, then check that the constant in the equation $Z=\omega$ given by $\nu_{\Point}$ satisfies conditions on the prefix, suffix and infixes given in the elements of $(\Prefix\tcartesian\Suffix)\treduced\FactorCode$. If yes, set the corresponding strings in elements of $\Prefix$, $\Suffix$, and $\FactorCode$ to $\omega$. Otherwise, reduce to $\bot$.
\item if the element of $\FactorCode$ does not contain factors existing in elements of $\Prefix$ or $\Suffix$ as subfactors, add them to the element of $\FactorCode$.
\end{itemize}

Henceforth we assume that any abstract string-interval lower bound $\infopat{\nu}\in\Point\treduced(\Prefix\tcartesian\Suffix)\treduced\FactorCode$ is reduced in the sense described above. Still, for the sake of brevity, in examples below we list uniformly equations determining $\Point$, $\Prefix$, $\Suffix$, and $\FactorCode$ elements. Since any equation $Z = X_i \omega_i Y_i$ is uniquely determined by $\omega_i$, we can represent elements of $\FactorCode$ by sets of strings $\omega_i$ when they are considered within $\FactorCode$ domain, while, when used in a product with $\Prefix$ and $\Suffix$ domains, they are unfolded into sets of corresponding word equations. 

Given $\nu_1,\nu_2\in \Prefix$, we define their join and meet operations as follows. 

$$\begin{array}{c}\nu_1 \supr_{\Prefix} \nu_2 = \begin{cases}\;\nu_i,\,\text{if } \nu_{2-i}=\bot,\\ 
\begin{array}{l}\{Z=\omega_0 X_0\},\text{ if } \\\quad\nu_1 = \{Z=\omega_0\delta_1\omega_1 X_0\},\,\nu_2 = \{Z=\omega_0\delta_2\omega_2 X_0\},\delta_1\neq \delta_2,\end{array}\\ 
\;\nu_i,\,\text{if }\nu_i = \{Z=\omega_0 X_0\},\,\nu_{2-i} = \{Z=\omega_0\omega_1 X_0\}.\end{cases} 
\\
\nu_1 \infn_{\Prefix} \nu_2 = \begin{cases}\bot,\,\text{if } \nu_1=\bot\logor \nu_2 = \bot,\\ \bot,\,\text{if } \nu_1 = \{Z=\omega_0\delta_1\omega_1 X_0\},\nu_2 = \{Z=\omega_0\delta_2\omega_2 X_0\},\delta_1\neq \delta_2,\\ 
\nu_{2-i},\,\text{if } \nu_i = \{Z=\omega_0 X_0\},\,\nu_{2-i} = \{Z=\omega_0\omega_1 X_0\}.\end{cases}\end{array}$$

In the similar way, the lattice operations on $\Suffix$ are defined~\cite{CostantiniFerrara,Mastroieni}. 

Given factor codes $I_1$ and $I_2$ represented by sets of strings, the operation $I_1\supr_{\FactorCode} I_2$ constructs the set of all their common substrings of maximal length. The operation $I_1\infn_{\FactorCode} I_2$ constructs $I_1\cup I_2$ and then filters out any string that is a proper substring of another element in the resulting set, in order to construct a valid factor code. The associativity and absorption for the given $\supr_{\FactorCode}$, $\infn_{\FactorCode}$ hold due to the corresponding properties of set operations, hence $\supr_{\FactorCode}$ and $\infn_{\FactorCode}$ are intersection and union of all the subwords of the factor codes of their arguments. The lattice $\FactorCode$ satisfies also the finite ascending chain property: indeed, $\supr_{\FactorCode}$ decreases set of unavoidable words determined by factor codes of its arguments.

If a string set contains $\empt$, the lower bound of the string interval must be trivial, since $\empt$ belongs to the only ideal generated by itself. Taking into account the fact that $\empt$ can be a result of out-of-bounds operation (such as $\osubstr{\omega}{|\omega|}$), considering such lower bounds looks like an over-generalization. On the other hand, $\empt$ satisfies any non-trivial disequality of the form given in~Fig.~\ref{Fig::basicint}, so it cannot be separated from other strings by any given interval upper bound. Tracking the fact that the given parameter can take an empty word as a value is useful, since $\empt$ possesses unique properties as a unit in the string monoid: it occurs in any position in any word. Therefore, the empty word can be seen as a ``singular point'' of a string set, and is desirable to be tracked separately. 

Now we are ready to define the abstract string property formally.

\begin{definition}

A \emph{string-property} abstract domain $\StringProperty$ is a combined product of the form $\Bool\tcartesian\bigl(\Point\treduced(\Prefix\tcartesian\Suffix)\treduced \FactorCode\treduced \AntiDictionary\bigr)$ where, given a value $\absval\in\StringProperty$:

\begin{itemize}
\item $\Bool$ is a trivial lattice tracking if $\empt\in \concr{\absval}$;
\item $\Point$ is the lattice tracking if $\concr{\absval}\setminus \{\empt\}$ is a singleton, and its value;
\item $\Prefix$ is the lattice tracking common prefix of values in $\concr{\absval}\setminus \{\empt\}$;
\item $\Suffix$ is the lattice tracking common suffix of values in $\concr{\absval}\setminus \{\empt\}$; 
\item $\FactorCode$ is the lattice tracking common factors of values in $\concr{\absval}\setminus \{\empt\}$;
\item $\AntiDictionary$ is the lattice tracking forbidden factors of values in $\concr{\absval}$.
\end{itemize}

We assume that if the underlying alphabet is not unary, $\AntiDictionary$ is trivial, hence, the corresponding component of the product is omitted by default.

Given values $\nu,\nu'\in\StringProperty$ in a non-unary alphabet, where $\nu=\biggeral\nu_{\Bool},\langle\nu_{\Point},\nu_{\Prefix},\nu_{\Suffix},\nu_{\FactorCode}\rangle\biggerar$, $\nu'=\biggeral\nu'_{\Bool},\langle\nu'_{\Point},\nu'_{\Prefix},\nu'_{\Suffix},\nu'_{\FactorCode}\rangle\biggerar$ and $\latop\in\bigl\{\supr, \infn\bigr\}$, $\nu\latop_{\StringProperty}\nu'$ is defined as $\biggeral\nu_{\Bool}\latop_{\Bool}\nu'_{\Bool}, \rho\bigl(\nu_{\Point}\latop_{\Point} \nu'_{\Point}, \nu_{\Prefix}\latop_{\Prefix} \nu'_{\Prefix},\nu_{\Suffix}\latop_{\Suffix} \nu'_{\Suffix},\nu_{\FactorCode}\latop_{\FactorCode} \nu'_{\FactorCode}\bigr)\biggerar$.

A property in a unary alphabet is represented by a value in $\Bool\tcartesian \bigl(\mathcal{LB}\treduced\mathcal{UB}\bigr)$, where $\Bool$ is the trivial lattice tracking $\empt$ and $\mathcal{LB}\treduced\mathcal{UB}$ is the abstract numeric domain of integer intervals with the positive closed lower bound\footnote{This domain is already discussed in Example~\ref{Example::IntIntervalReduced}.}.  
\end{definition}

\begin{figure}[h]\centering
$
\begin{array}{ccc}
\raisebox{7.5ex}{\begin{tikzpicture}[scale=0.85]
\node[circle,draw=black!50!gray,thick,fill=myBlue!20!white,inner sep=1ex] (1) at (0,0) {$\bot$};
\node[circle,draw=black!50!gray,thick,fill=myBlue,inner sep=1ex] (2) at (0,1.8) {$\top$};

\draw[-,ultra thick,gray] (1) -> (2);
\node at (0.6,0.9) {$\times$};

\node(shapedlat) at (2,0.9) {\eqinterval{1.1}{}};

\node(shapeulat) at (2,1.8) {\neqinterval{1.1}{}};

\node(shapeblat) at (2,0){\syminterval{1.1}{}};

\draw [very thick,decorate, 
	decoration = {brace,
		raise=0ex,
		amplitude=2ex}] (1.35,-0.5 ) --  (1.35,2.25) node[pos=0.5,right=-2ex,yshift=6ex,black!50!gray]{};
\end{tikzpicture}}
&\quad&
\begin{tikzpicture}[scale=1.12]
\node[circle,draw=black!50!gray,thick,fill=myBlue!20!white,inner sep=0.5ex,minimum width=7.5ex] (1) at (1,0) {$\bot$};
\node[circle,draw=black!50!gray,thick,fill=myBlue!20!white,inner sep=0.5ex,minimum width=7.5ex] (2) at (0,1) {$\{a\}$};
\node[circle,draw=black!50!gray,thick,fill=myBlue!20!white,inner sep=0.2ex,minimum width=7.5ex] (3) at (2,1) {$\mathsmaller{\mathsmaller{\Sigma^* b \Sigma^*}}$};
\node[circle,draw=black!50!gray,thick,fill=myBlue!20!white,inner sep=0.75ex,minimum width=7.5ex] (4) at (1,2) {$\mathsmaller{\Sigma^*\setminus\{\empt\}}$};

\node (eqB) at (2,3.5) [draw=black!30!gray,line width=0.2ex,rectangle callout,shape border rotate=0, minimum height=2ex,inner sep = 0.4ex,fill=white, callout relative pointer={(-0,-2.3)},callout pointer width=3ex] {\small\begin{tabular}{c}
					\textcolor{black!30!gray}{$\wequat{Z}{XbY}$}\end{tabular}};

\node (eqA) at (-0.25,-0.35) [draw=black!30!gray,line width=0.2ex,rectangle callout,shape border rotate=0, minimum height=2ex,inner sep = 0.4ex,fill=white, callout relative pointer={(0.25,0.85)},callout pointer width=3ex] {\small\begin{tabular}{c}
					\textcolor{black!30!gray}{$\wequat{Z}{a}$}\end{tabular}};

\node[circle,draw=black!50!gray,thick,fill=myBlue,inner sep=0.5ex,minimum width=7.5ex] (11) at (-2,1) {${\{\varepsilon\}}$};
\node[circle,draw=black!50!gray,thick,fill=myBlue,inner sep=0.1ex,minimum width=7.5ex] (12) at (-3,2) {$\mathsmaller{{\{\empt, a\}}}$};
\node[circle,draw=black!50!gray,thick,fill=myBlue,inner sep=-0.3ex,minimum width=7.5ex] (13) at (-1,2) {$\begin{array}{l}\mathsmaller{\mathsmaller{\Sigma^*b\Sigma^*}}\\\mathsmaller{\,\cup\,\{\varepsilon\}}\end{array}$};
\node[circle,draw=black!50!gray,thick,fill=myBlue,inner sep=0.5ex,minimum width=7.5ex] (14) at (-2,3) {$\mathlarger{\Sigma^*}$};

\node (eqBE) at (-0.3,3.5) [draw=black!30!gray,line width=0.2ex,rectangle callout,shape border rotate=0, minimum height=2ex,inner sep = 0.4ex,fill=white, callout relative pointer={(-0.7,-1.1)},callout pointer width=3ex] {\small\begin{tabular}{c}
					\textcolor{black!30!gray}{$\wequat{ZXbY}{XbYZ}$}\end{tabular}};

\node (eqAE) at (-2.5,-0.35) [draw=black!30!gray,line width=0.2ex,rectangle callout,shape border rotate=0, minimum height=2ex,inner sep = 0.4ex,fill=white, callout relative pointer={(-0.5,2.15)},callout pointer width=3ex] {\small\begin{tabular}{c}
					\textcolor{black!30!gray}{$\wequat{ZX}{a}$}\end{tabular}};

\draw[-,ultra thick,gray] (1) -> (2);
\draw[-,ultra thick,gray][-] (1) -> (3);
\draw[-,ultra thick,gray][-] (2) -> (4);
\draw[-,ultra thick,gray][-] (3) -> (4);

\draw[-,ultra thick,gray] (11) -> (12);
\draw[-,ultra thick,gray][-] (11) -> (13);
\draw[-,ultra thick,gray][-] (12) -> (14);
\draw[-,ultra thick,gray][-] (13) -> (14);

\begin{pgfonlayer}{background}
\draw[-,ultra thick,gray] (1) -> (11);
\draw[-,ultra thick,gray] (2) -> (12);
\draw[-,ultra thick,gray] (3) -> (13);
\draw[-,ultra thick,gray] (4) -> (14);
\end{pgfonlayer}
\end{tikzpicture}

\end{array}
$
\caption{A string property as a Cartesian product}
\label{Fig::strProp}
\end{figure}
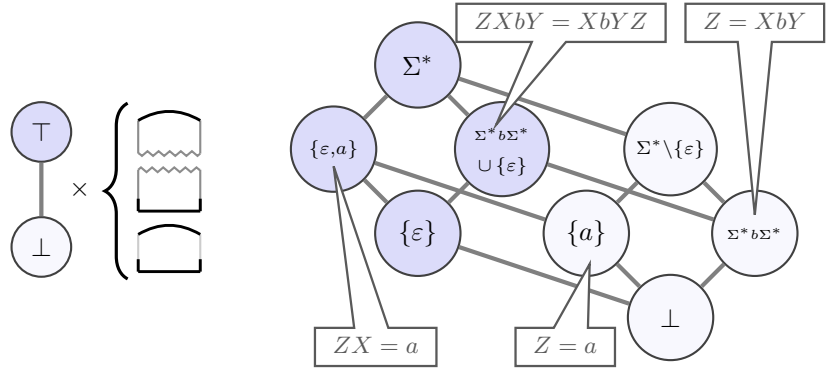

Fig.~\ref{Fig::strProp} shows an example of a lattice describing a non-unary string property. Note that the concretisation sets of the ``upper sublattice'' (shown in saturated blue) describing possibly empty strings can also be defined in terms of word equations. In order to describe possible ``equation primitives'', we studied instances of all basic 3-variables' equations given in the paper \cite{MakaninZoo1}~by G.~Makanin, having elementary parametric solutions. The results of the study are given in the Appendix, Subsection~\ref{subsect::dense}. Actually, all the primitively specialized equations by Makanin define languages described by means of equation systems of the forms presented in Fig.~\ref{Fig::basicint}.

\begin{example}
Consider $\nu_1 = \langle \top, \wequat{Z}{XabY}\rangle$, $\nu_2 = \langle \bot, \wequat{Z}{XbaY}\rangle$, then $\nu_1\supr\nu_2 = \langle \top, \{\wequat{Z}{XaY},\wequat{Z}{XbY}\}\rangle$.
\end{example}
  
\section{Abstract String Objects}\label{Sect:object}

Let $\Sigma$ be a set of Unicode characters and let us consider a string morphism $\lenmorph$ s.t. $\forall \delta\in \Sigma \bigl(\lenmorph(\delta)=\nullsymb\bigr)$, where $\nullsymb$ denotes the character with the code $0$. Such a morphism maps any string to unary representation of its length\footnote{We assume that invisible characters, including backspace \texttt{'\textbackslash{}b'}, are all counted in the length, according to \JS{} semantics.}. String operations, when considered in the unary alphabet $\{\nullsymb\}$, become usual arithmetic operations, such as addition, subtraction in positive domain, and 
comparison. Hence, we can understand string length as a string property preserved by the morphism $\lenmorph$. This approach can be generalized to length-non-increasing morphisms other than $\lenmorph$. In this scope, any string can seen as a complex object possessing a couple of properties.

\begin{itemize}
\item The string value is its property preserved by the identity morphism.
\item The length of the string is its property preserved by morphism $\lenmorph$ that glues all letters in $\Sigma$ into a single equivalence class.
\item Other string properties can be tracked by length-non-increasing morphisms defined via a partition of $\Sigma$ to equivalence classes, whose elements are mapped to single letters or $\empt$.
\end{itemize}

This concept generalizes the concept of a \JS{} string object, possessing value and length properties, in an abstract manner. Not any morphism is helpful for tracking string invariants. For example, given $\sigma_{\tiny \begin{array}{l}a\mapsto b\\ b\mapsto a\end{array}}(\delta) = \begin{cases}a, \delta = b,\quad b, \delta = a\\ \delta,\text{otherwise}\end{cases}$\hspace*{-2ex}, it captures exactly the same properties as the value, being bijective on $\Sigma$. Similarly, tracking properties preserved by the  morphism mapping all elements of $\Sigma$ to $b$ is meaningless since these properties are already captured by the length. Hence, we are interested in the morphisms $\sigma$ determined by partitions of $\Sigma$, say $\Sigma_1,\dots, \Sigma_k, \Sigma_{\empt}$ with $|\Sigma_i|>1$ s.t. $\sigma$ is idempotent, non-increasing wrt alphabetic order, and preserves the partition, i.e. $\forall\delta\in\Sigma_i\bigl(\sigma(\delta)\preceq \delta \logand \sigma(\delta)\in\Sigma_i\bigr)\logand \forall \delta\in\Sigma_{\empt}\bigl(\sigma(\delta)=\empt\bigr)$. Here $\preceq$ stands for the code ordering on Unicode characters, hence, $\sigma(\delta_i)$ in the formula above is the character in $\Sigma_i$ with the minimal code. 
We tacitly assume that any partial morphism defined on $\biggerl\bigcup_{i=1}^k \Sigma_i\cup \Sigma_{\empt}\biggerr\subset \Sigma$ is extended to a global morphism on $\Sigma$ via the identity map, which represents the partition of $\Sigma\setminus\biggerl\bigcup_{i=1}^k \Sigma_i\cup \Sigma_{\empt}\biggerr$ to singletons. We call morphisms of the form above \emph{standard} and assume by default that all the morphisms considered further satisfy the standard form. Note that, given any length-non-increasing morphism $\theta$ s.t. $\forall \delta\bigl(\theta(\delta)\preceq \delta\bigr)$, its fixpoint $\theta^*$ is a standard morphism. 

Morphisms are extended to equations and their systems as follows, assuming that any morphism $\sigma$ is identity on $\VarSet$.

\begin{itemize}
\item If $\sigma$ is non-erasing on letters in $\mathcal{V}$, then $\sigma\bigl(\wequat{Z}{\mathcal{V}}\bigr) = \biggerl\wequat{Z}{\sigma(\mathcal{V})}\biggerr$. 
\item Given an equation $\wequat{Z}{\omega_0 \upsilon_1 \omega_1 \dots \upsilon_k \omega_k}$, $\omega_i\in\bigl(\Sigma\setminus \Sigma_{\empt}\bigr)^+$, $\upsilon_i\in \Sigma_{\empt}^+$, decompose it to the system $\begin{cases}\wequat{Z}{\sigma(\omega_0) Y_0},\, \wequat{Z}{X_0 \sigma(\omega_k)}\\ \bigcup_{i=1}^{k-1}\biggerl\wequat{Z}{X_i \sigma(\omega_i) Y_i}\biggerr.\end{cases}$\hspace*{-1.5ex}Suffix and prefix equations are treated in the same way.
\end{itemize}

Given a string lower bound $\absval=\langle \absval_{\Bool}, \infopat{\absval}\rangle$, where $\infopat{\absval}\hspace*{-1.25ex}\in \Point\treduced \bigl(\Prefix\tcartesian\Suffix\bigr)\treduced\FactorCode$, 

$\sigma(\absval)=\begin{cases}\langle\top,\bot\rangle,\text{if }\infopat{\absval}=\bot\text{ or }\sigma(\absval_{\Point})=\empt,\\\langle \top, \top\rangle,\text{if }\Sigma_{\empt}(\sigma)\neq \varnothing,\text{and }\sigma(\absval_{\Prefix}),\sigma(\absval_{\Suffix}),\sigma(\absval_{\FactorCode})\text{ are equal to }\top,\\ \biggeral\absval_{\Bool}, \absval_{\Point}\treduced\bigl(\sigma(\absval_{\Prefix})\tcartesian\sigma(\absval_{\Suffix})\bigr)\treduced \sigma(\absval_{\FactorCode})\biggerar,\text{ otherwise}.\end{cases}$

Considering all the above, we can give a formal definition of the abstract string object domain based on string properties. Given a morphism $\sigma$, we denote by $\StringProperty(\sigma)$ the string property preserved by the morphism $\sigma$, i.e. abstracting any concrete string $\omega$ to $\abstr{\sigma(\omega)}$.

\begin{definition}
An abstract string object domain $\StringObjectDom$ is $$\StringProperty(id)\treduced \StringProperty(\lenmorph)\treduced \raisebox{-0.75ex}{$\Reduced{1.2}$}_{i=1}^n\StringProperty(\sigma_i)$$ where $\StringProperty(id)$ is the value domain, $\StringProperty(\lenmorph)$ is the length domain, and $\raisebox{-0.75ex}{$\Reduced{1.2}$}_{i=1}^n\StringProperty(\sigma_i)$ is an optional reduced product of other properties, defined by the standard morphisms $\sigma_i$. 

Hence, any abstract string object $\absval$ can be represented in the form $$\objfullreprN{\absval}{1}{n}.$$   

\end{definition}

By default, we assume that all the string properties defined by $\sigma_i$ generating strings in non-unary alphabets contain only lower bounds of string intervals. We recall that if a string property is defined by a morphism onto the unary alphabet (e.g. the length property), it is presented by an element of the product $\Bool\tcartesian \bigl(\mathcal{LB}\treduced\mathcal{UB}\bigr)$.

Generally, custom properties are taken from the program to be analysed. Replacement and inclusion operators taking character sets as their arguments are the main sources of the alphabet partitions supported by the morphisms. We even can create the properties dynamically during abstract interpretation; in this case, we should introduce join and meet operations for \textit{sets of properties}, as well as the appropriate reduction strategy. This feature is discussed in more detail in Sect.~\ref{Sect:distinct}. We postpone the proof that $\StringObjectDom$ is a lattice until the reduction function $\rho_{\StringObjectDom}$ is defined.

\section{String Object Reduction}\label{Sect:reduction}

We aim the reduction function $\rho$ to construct an equivalent representation of a given object, which is assumed to capture the concretisation set of the object in succinct way. 

String object representations can be reduced in the following ways:
\begin{itemize}
\item reduction of a single property of the object;
\item crossover reduction of the object properties.
\end{itemize}

We show that for linear word equations, the reduced form of a single property is unique and depends on the output alphabet of the property morphism. 
Moreover, while a crossover reduction of the properties product is NP-complete even if only the length and value properties are considered, there are practically useful sets of string objects and string properties for which the crossover reduction can be performed efficiently.

\subsection{Standalone Reduction of String Properties}

\begin{definition}
Let us consider an equation system $\infopat{\absval}$ determined by a product in $(\Prefix\tcartesian\Suffix)\treduced\FactorCode$. We say that the word $\upsilon$ is an unavoidable word in $\concr{\infopat{\absval}}$, or unavoidable with respect to the system given by $\infopat{\absval}$, if $\upsilon$ occurs in all words from $\concr{\infopat{\absval}}$ as a subword.

We say that the value $\infopat{\absval}$ is reduced if, when some word $\omega$ is unavoidable in the set $\concr{\infopat{\absval}}$, then $\omega$ is a subword of some element of $\infopat{\absval}$. 

\end{definition}

The equation reduction procedure varies depending on the alphabet $\Sigma$ of the equation systems solutions.

Let $\infopat{\absval}\in (\Prefix\tcartesian \Suffix)\treduced\FactorCode$ be an equation system consisting of the word equations $\wequat{Z}{\upsilon_1Y_0},\wequat{Z}{X_i\omega_i Y_i},\wequat{Z}{X_0 \upsilon_2}$, such that no $\omega_i$ is a substring of $\upsilon_1$, $\omega_j$, or $\upsilon_2$. Let $S=\Sol_Z\bigl(\infopat{\absval}\bigr)$. We say that an unavoidable in $S$ word $\upsilon$ violates the reduced-form condition, if $\upsilon$ is not a subword of any of the words $\nu_1$, $\nu_2$, $\omega_1$,\dots, $\omega_n$.

\begin{restatable}{lemma}{lemsinglereduction}
\label{Lemma:standaloneReduction}
\begin{itemize}
\item If $|\Sigma|>2$, then in $S$ there are no unavoidable words violating the reduced-form condition.
\item If $|\Sigma|=\{a,b\}$, then any unavoidable in $S$ word violating the reduced-form condition takes only one of the following forms: $a^k b$, $a b^k$, $b^k a$, or $b a^k$, where $k\geq 1$, and can be found in $\bigO(n\cdot\log{n})$ time, where $n$ is the number of equations in $\infopat{\absval}$.
\end{itemize}
\end{restatable}

The proof is given in Appendix. The main idea is to create appropriate $\omega$-\textit{delimiters} between words occurring in $\infopat{\absval}$ in order to construct a counterexample: a word satisfying $\infopat{\absval}$ but not including $\omega$. Given $\omega\in\Sigma^+$, an $\omega$-delimiter is a word $\tau$ s.t., if $\omega$ occurs in a string $\omega_1 \tau \omega_2$, then $\omega$ does not overlap either with $\omega_1$ or with $\omega_2$. For example, for $\omega = a^k$, the string $\tau = b$ can serve as the $\omega$-delimiter; for $a^2 b^2$, we can choose the delimiter $abab$. Both these examples demonstrate perfect delimiters, because $\omega$ and $\tau$ do not overlap as well. Given $\omega=aab$, it has no perfect $\omega$-delimiter in $\Sigma=\{a,b\}$. That is why this word can violate the reduced-form condition.

\begin{example}
Consider the alphabet $\{a,b\}$.
\begin{itemize}
\item Given the set of all strings including substrings $\{abaa, bbaa\}$, the word $aab$ is unavoidable in them, hence, value $\{\wequat{Z}{X_1 abaa Y_1},\wequat{Z}{X_2 bbaa Y_2}\}$ reduces to $\{\wequat{Z}{X_1 abaa Y_1},\wequat{Z}{X_2 bbaa Y_2},\wequat{Z}{X_3 aab Y_3}\}$ in the given alphabet.
\item Given the set of equations $\{\wequat{Z}{X_1 a^3 b^2 Y_1},\wequat{Z}{X_2 a^2 b^3 Y_2}\}$, the word $aab$ is not unavoidable in them, given the word $b^3 a^3$ satisfying the both equations but not including $aab$.
\end{itemize}
\end{example}

\subsection{Upper Bound Problem in String Object Reduction}

The reduction of a single string property is a simple task, but their cross-reduction may explode complexity of the string objects processing. The main reason is existence of string-intervals' upper bounds. In this section mostly concentrate on crossover reduction of the main two properties of string objects: the value wrt the length, and the length wrt the value. 

Given a string object $\absval$, we say its value is perfectly reduced iff it contains a factor code with all unavoidable subwords of the set $\concr{\absval}$, and captures its maximal common prefix and suffix. The length of $\absval$ is (perfectly) reduced, iff its lower and upper bounds are strict in $\concr{\absval}$.

\begin{example}
Let $\objval{\absval}=\begin{cases}\wequat{Z}{X_1 abbab Y_1}\\\wequat{Z}{X_2 abab Y_2}\end{cases}$, $\objlen{\absval}=[6;8)$. Neither the value nor the length of $\absval$ is reduced. Indeed, no word of the length 6 includes both $abbab$ and $abab$. The appropriate words of the length 7 are $abbabab$ and $ababbab$. Hence, $\concr{\absval}=\{abbabab,ababbab\}$. These words also have common suffix $bab$ and prefix $ab$. Hence, the value $\absval$ is reduced to 
$\SOb{\begin{cases}\wequat{Z}{ab Y_0},\; \wequat{Z}{X_0bab}\\\wequat{Z}{X_1 abbab Y_1},\;\wequat{Z}{X_2 abab Y_2}\end{cases}}{7}{8}$. 
\end{example}

Generally, the perfect reduction of the given abstract object is NP-hard. Indeed, let the object value be $\begin{cases}\wequat{Z}{X_1\omega_1 Y_1}\\\dots\\\wequat{Z}{X_k \omega_k Y_k}\end{cases}$, and the length be represented by $\bigl[1, \sum_{i=1}^k \omega_i+1\bigr)$. Hence, the perfect reduction is to determine the shortest common superstring of $\omega_1,\dots,\omega_k$, while the task is known to be NP-complete~\cite{Kulikov}. Therefore, we determine some practically useful cases when the perfect reduction is simple, and use a partial reduction strategy in other cases. In order to guarantee satisfiability of the lattice axioms, we show that the partial reduction results are stable wrt them.

\begin{restatable}{lemma}{lemlenreduction}
\label{Lemma:ReducedLength}
Let $\objval{\absval}$ be $\begin{cases}\wequat{Z}{\upsilon_0 Y},\;\wequat{Z}{X\upsilon_1},\\\wequat{Z}{X_1\omega_1 Y_1},\dots,\wequat{Z}{X_k\omega_k Y_k}\end{cases}$. If the infinum of $\objlen{\absval}$ is at least $\biggerl\displaystyle\sum_{i=1}^k |\omega_i|\biggerr +|\upsilon_0|+|\upsilon_1|+k-1+\min(|\upsilon_0|,1) + \min(|\upsilon_1|,1)$, then both the value and the length of $\absval$ are already reduced wrt each other.
\end{restatable}

The general outline of the proof (whose complete version is given in Appendix) again uses the delimiter concept. For the value property, assuming that the alphabet $\Sigma$ is large enough, a universal delimiter $\delta$ is chosen that does not occur explicitly in $\objval{\absval}$. Interspersing this delimiter between constants of $\objval{\absval}$ results in a counterexample that cannot include any string not already specified in $\objval{\absval}$. The lower bound on the length given in Lemma~\ref{Lemma:ReducedLength} is strict, which can be shown by the following example.

\begin{example}
The equation systems $\{\wequat{Z}{a Y},\wequat{Z}{X_1 ba Y_1}\}$ and $\{\wequat{Z}{X_1 ab Y_1},\wequat{Z}{X a}\}$, given the length $3$, are concretised to the same set $\{aba\}$. Hence, their reduced representation for the length $3$ is $\wequat{Z}{aba}$.
\end{example}

\subsection{Cross-Reduction of String Object Lower Bounds}

While the task of determining a strict lower bound of $\objlen{\absval}$ is hard, some constraints on it are implied from $\objval{\absval}$ almost trivially. Indeed, given any constant string $\omega$ occurring in $\objval{\absval}$, the abstract object length cannot be less than $|\omega|$. Hence, the lower bound length constraint imposed on $\absval$ is not less than $\lenmorph\bigl(\objval{\absval}\bigr)$. This observation demonstrates a general idea of string properties ordering.

We denote a set of $\omega$ preimages wrt a morphism $\sigma$ with $\sigma^{-1}(\omega)$. If $\sigma^{-1}(\omega)$ is a singleton, then the notation is overloaded to denote its element.

\begin{definition}\label{defin:morphord}
Given standard morphisms $\sigma_1$ and $\sigma_2$, we say that $\sigma_1\preceq \sigma_2$ iff $\sigma_2\circ \sigma_1 = \sigma_2$. Similarly, $\StringProperty(\sigma_1)\preceq\StringProperty(\sigma_2)$ iff $\sigma_1\preceq\sigma_2$.

We say $\sigma_2$ preserves $\omega\in\Sigma^+$ wrt $\sigma_1$ iff:
\begin{itemize}
\item $\sigma_1(\omega)=\sigma_2(\omega)=\omega$;
\item $\sigma^{-1}_2(\omega)=\sigma^{-1}_2(\omega)$, or $|\omega|=1$ and the sets $\bigl\{\omega'\mid \omega'\in\sigma^{-1}_1(\omega)\logand |\omega'|=1\bigr\}$ and $\bigl\{\omega'\mid \omega'\in\sigma^{-1}_2(\omega)\logand |\omega'|=1\bigr\}$ coincide.
\end{itemize}
\end{definition}

According to Def.~\ref{defin:morphord}, the top property is determined by the trivial morphism mapping everything to $\empt$, and has a unique element $\langle\top,\bot\rangle$, and the minimal property is the value. A morphism partial ordering example is shown in Fig.~\ref{Fig::morphord}. The grey ``propagating'' arrows are explained later.

When $\sigma$ preserves $\omega$ wrt the identity morphism, we say that $\sigma$ preserves $\omega$. In this case, $\sigma^{-1}(\omega)=\{\omega\}$. Any erasing morphism can preserve at most one-letter words. 

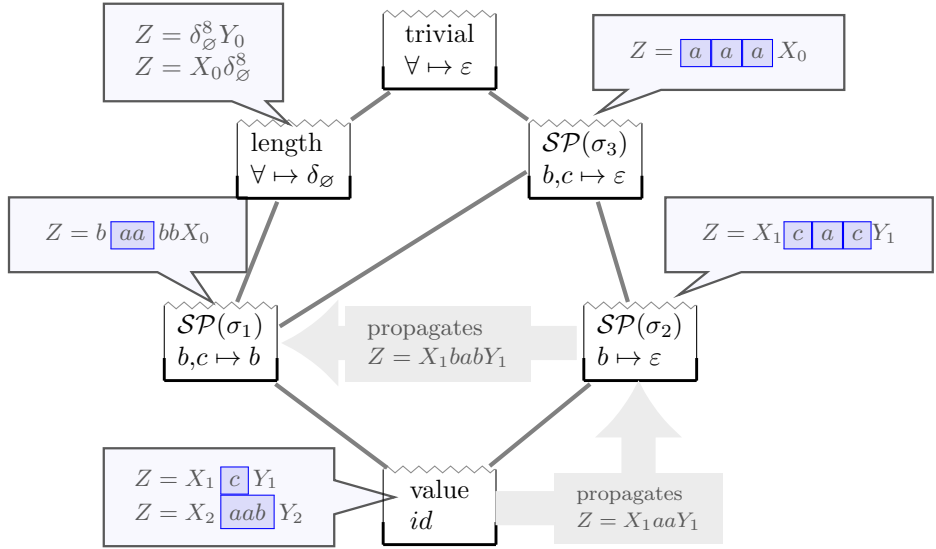
\begin{figure}[!htb]
\centering\begin{tikzpicture}[scale=0.97]

\node (valA) at (-3.25,-2.1) [draw=black!30!gray,text=black!30!gray,line width=0.2ex,rectangle callout,shape border rotate=0, minimum height=3ex,inner sep = 1.5ex,fill=myBlue!20!white, callout relative pointer={(0.8,0)},callout pointer width=3ex] {\small $\begin{array}{l}Z = X_1\,\fcolorbox{blue}{myBlue}{$c$}\,Y_1 \\
Z = X_2\,\fcolorbox{blue}{myBlue}{$aab$}\,Y_2\end{array}$};

\draw (-1,-2.75) pic{eqnode ={val,$\begin{array}{l}\text{value}\\id\end{array}$,,2,1}};


\node (valA1) at (-4.5,1.5) [draw=black!30!gray,text=black!30!gray,line width=0.2ex,rectangle callout,shape border rotate=0, minimum height=3ex,inner sep = 2ex,fill=myBlue!20!white, callout relative pointer={(0.5,-0.4)},callout pointer width=3ex] {\small $\begin{array}{l}Z=b\,\fcolorbox{blue}{myBlue}{$aa$}\,bbX_0\end{array}$};

\draw (-4,-0.5) pic{eqnode ={h1,$\begin{array}{l}\StringProperty(\sigma_1)\\  b{,}c\mapsto b\end{array}$,,2,1}};

\draw (1.75,-0.5) pic{eqnode ={h2,$\begin{array}{l}\StringProperty(\sigma_2)\\  b\mapsto \empt\end{array}$,,2,1}};


\node (valA2) at (4.7,1.5) [draw=black!30!gray,text=black!30!gray,line width=0.2ex,rectangle callout,shape border rotate=0, minimum height=3ex,inner sep = 2ex,fill=myBlue!20!white, callout relative pointer={(-0.6,-0.3)},callout pointer width=3ex] {\small $\begin{array}{l}
Z=X_1 \fcolorbox{blue}{myBlue}{$c$}\fcolorbox{blue}{myBlue}{$a$}\fcolorbox{blue}{myBlue}{$c$}Y_1\end{array}$};

\draw (1,2) pic{eqnode ={h3,$\begin{array}{l}\StringProperty(\sigma_3)\\  b{,}c\mapsto \empt\end{array}$,,2,1}};


\node (valA3) at (3.6,4) [draw=black!30!gray,text=black!30!gray,line width=0.2ex,rectangle callout,shape border rotate=0, minimum height=3ex,inner sep = 2ex,fill=myBlue!20!white, callout relative pointer={(-0.6,-0.3)},callout pointer width=3ex] {\small $\begin{array}{l}
Z=\fcolorbox{blue}{myBlue}{$a$}\fcolorbox{blue}{myBlue}{$a$}\fcolorbox{blue}{myBlue}{$a$}\,X_0\end{array}$};

\draw (-3,2) pic{eqnode ={len,$\begin{array}{l}\text{length}\\  \forall\mapsto \nullsymb\end{array}$,,2,1}};


\node (valA4) at (-3.6,4) [draw=black!30!gray,text=black!30!gray,line width=0.2ex,rectangle callout,shape border rotate=0, minimum height=3ex,inner sep = 1.5ex,fill=myBlue!20!white, callout relative pointer={(0.4,-0.3)},callout pointer width=3ex] {$\begin{array}{l}Z = \nullsymb^8 Y_0\\ Z= X_0\nullsymb^8\end{array}$};

\draw (-1,3.5) pic{eqnode ={triv,$\begin{array}{l}\text{trivial}\\\forall\mapsto\empt\end{array}$,,2,1}};


\draw[-,ultra thick,gray] (val) -> (h1);
\draw[-,ultra thick,gray][-] (val) -> (h2);
\draw[-,ultra thick,gray][-] ([xshift=5.35ex,yshift=1.5ex]h1.center) -> ([xshift=-5.35ex,yshift=-1ex]h3.center);
\draw[-,ultra thick,gray][-] (h1) -> (len);
\draw[-,ultra thick,gray][-] (h2) -> (h3);
\draw[-,ultra thick,gray][-] (h3) -> ([xshift=4.5ex,yshift=-3.7ex]triv.center);
\draw[-,ultra thick,gray][-] (len) -> ([xshift=-4.5ex,yshift=-3.7ex]triv.center);

 \draw[line width=3ex, gray!15,rounded corners=0.2ex, -{Latex[length=5ex, width=8ex]}] (0.54,-2.25) -- (2.5,-2.25) -- (2.5,-0.5);
 
\node (err) at (2.5,-2.25) [text=black!30!gray, fill=gray!15] {\footnotesize $\begin{array}{l}
\text{propagates}\\
Z=X_1 aa Y_1\end{array}$};

 \draw[line width=3ex, gray!15,rounded corners=0.2ex, -{Latex[length=5ex, width=8ex]}] ([xshift=-5.75ex]h2.center) -> ([xshift=5.5ex]h1.center);

\node (propagate) at (-0.25,0) [text=black!30!gray,fill=gray!15] {\small $\begin{array}{l}\text{propagates}\\
Z=X_1 bab Y_1\end{array}$};
\end{tikzpicture}
\caption{\small Morphisms ordering, words preserved by the morphisms (given in blue boxes), and equation propagation via reduction. Only the lower bounds are considered, and the length lower bound is given as a set of equations for uniformity.}\label{Fig::morphord}
\end{figure}

\begin{restatable}{lemma}{lempropagation}\label{Lemma:Propagation}
Given string lower bounds $\infopat{\absval}_1\in\StringProperty(\sigma_1)$ and $\infopat{\absval}_2\in\StringProperty(\sigma_2)$ s.t. $\sigma_2 \succ \sigma_1$ and they are non-erasing, all equations that can be propagated from $\infopat{\absval}_2$ to $\infopat{\absval}_1$ include words preserved by $\sigma_2$ wrt $\sigma_1$.
\end{restatable}

The proof uses an idea of sliding counterexample construction: if $\omega$ does not occur in $\infopat{\absval}_1$ and is not preserved in $\infopat{\absval}_2$, then, given any $\upsilon=\delta_0\dots \delta_k$ in $\infopat{\absval}$ there exists a position $i<|\omega|$ s.t. $\sigma_1\bigl(\sigma^{-1}_2(\delta_i)\bigr)$ mismatches with $i$-th letter of $\omega$. Hence, we slide a possible starting position of $\omega$ in an element of $\sigma^{-1}(\upsilon)$ until all length of $\upsilon$ is exhausted, thus constructing a counterexample to the initial assumption that $\omega$ is unavoidable in $\sigma^{-1}(\upsilon)$. Using the sliding counterexample scheme together with the $\omega$-delimiter construction for binary alphabet (considered in Lemma~\ref{Lemma:standaloneReduction}) helps to derive one more useful statement.

\begin{restatable}{lemma}{lemmasharing}\label{Lemma::sharing}
If morphisms $\sigma_1$ and $\sigma_2$ are incomparable, an equation $\wequat{Z}{X\omega Y}$ does not occur in $\infopat{\absval}_1\in\StringProperty(\sigma_1)$, and $\sigma_1(\omega)$ is unavoidable in $\sigma_1\bigl(\sigma^{-1}_2(\infopat{\absval}_2)\bigr)$, $\infopat{\absval}_2\in\StringProperty(\sigma_2)$, then either $\omega$ is preserved by $\sigma_2$ wrt $\sigma_1$, or
 there exists $\delta$ s.t. $\sigma_2(\delta)=\empt$, $\sigma_1(\delta)=\delta$, and $\omega = \delta^{k_1}\omega'\delta^{k_2}$, where $\omega'$ is preserved by $\sigma_2$ wrt $\sigma_1$. 
\end{restatable}

We see that, given two lower bounds of a same string object $\absval$, $\infopat{\absval}_1\in\StringProperty(\theta_1)$ and $\infopat{\absval}_2\in\StringProperty(\theta_2)$, the ways to propagate unavoidable factors from $\infopat{\absval}_1$ to $\infopat{\absval}_2$ are rather limited.

\textbf{Factor propagation upwards}: if $\theta_1 \prec \theta_2$, then we can propagate the whole $\theta_2(\infopat{\absval}_1)$ to $\infopat{\absval}_2$. In Fig.~\ref{Fig::morphord}, this situation is shown given $\theta_1=id$, $\theta_2=\sigma_2$. The $\sigma_2$-image of the value is $\bigl\{\wequat{Z}{{X_1} c{Y_1}},\wequat{Z}{{X_2} aa{Y_2}}\bigr\}$. While $c$ already occurs in the factor code of $\infopat{\absval}_2$, the factor $aa$ is to be propagated.

\textbf{Preserved factor propagation}: if $\theta_1\not\prec\theta_2$, then we can propagate the factors preserved by $\theta_1$ from $\infopat{\absval_1}$ to $\infopat{\absval_2}$. If $\theta_1$ and $\theta_2$ are non-erasing, this is the only possible option. In Fig.~\ref{Fig::morphord}, we could propagate the factor $aa$ from $\StringProperty(\sigma_1)$ to $\StringProperty(\sigma_2)$ this way.

\textbf{Gap subwords propagation}: if $\theta_1$ is erasing and its factor code includes a word $\omega = \delta^i \omega' \delta^j$ s.t. $\forall \delta'\in\Sigma_{\empt}(\theta_1)\bigl(\theta_2(\delta')=\theta_2(\delta)\bigr)$, and $\omega$ is preserved wrt $\theta_2$, then $\theta_2(\omega)$ can be included in the factor code of $\infopat{\absval}_2$. This situation is shown via the grey arrow in Fig.~\ref{Fig::morphord} given $\theta_1=\sigma_2$, $\theta_2=\sigma_1$. Indeed, in the factor $cac$, $a$ is preserved by $\sigma_2$ wrt $\sigma_1$, and $\sigma_1(c)=\sigma_1\bigl(\Sigma_{\empt}(\sigma_2)\bigr)$. Hence, the unavoidable factor $cac$ wrt $\sigma_2$ corresponds to the unavoidable pattern $bb^* a b^*b$ wrt $\sigma_1$.

In factor codes cross-propagation, reduction to $\bot$ never occurs. 

Prefix, suffix, and constant equations occurring in the lower bounds can be propagated more eagerly. Namely, we can see them as sequences of constraints on positions in prefixes, suffixes, and strings occurring in $\concr{\absval}$, and resolve the constraints for all these positions. If the constraints are contradictory, the whole string object is reduced to $\bot$. 

\begin{example}
Given $\theta_1 = \bigl\{\{a,b\}\mapsto a, c\mapsto\empt\bigr\}$, $\theta_2 = \bigl\{\{a,c\}\mapsto a,\{b,d\}\mapsto b\bigr\}$, if an element of $\StringProperty(\theta_1)$ includes $\wequat{Z}{ad Y_0}$, and the element of $\StringProperty(\theta_2)$ includes $\wequat{Z}{ba Y_0}$, the constraints on the first two letters $\delta_1$, $\delta_2$ of the elements of $\concr{\absval}$ imposed by these equations are: $\begin{cases}\delta_1 \in\{a,b,c\}\logand \delta_1\in\{b,d\}\\ \delta_2\in\begin{cases}\{c,d\},\delta_1\neq c\\\{a,b,c\},\delta_1=c\end{cases}\logand\delta_2\in\{a,c\}.\end{cases}$ 

The prefix equation $\wequat{Z}{bcY_0}$ derived by resolving the constraints can be propagated to the value of $\absval$ satisfying these two lower bounds.
\end{example}

Hence, we have the outline of the lower bounds cross-reduction algorithm.
\begin{itemize}
\item Gather constraints on prefixes, suffixes, and constant values, and resolve them, propagating the resulting equations to string properties defined with least possible morphisms.
\item Perform downward propagation of all preserved subwords.
\item By breadth-first graph traversal, perform upward propagation, and generate all gap subwords, moving from the least property (e.g. from the value) upwards.
\end{itemize}

While this cross-reduction strategy is time-consuming, in Sect.~\ref{Sect:operations} we discuss why it can be partially skipped in computations in $\StringObjectDom$. The strategy has another limitation: it ignores a special structure of the lower bounds in unary alphabets. We fix this limitation by adding the last reduction step described in the rest of this section.

Let $\theta$ be a morphism on $\Sigma^*$ with a unary image alphabet. We say $\sigma$ is compatible with $\theta$ iff $\Sigma_{\empt}(\theta)\subseteq\Sigma_{\empt}(\sigma)$. We aim at constructing the lower bound $n$ on the length of words in an element of $\StringProperty(\theta)$. The state-of-art unary bounds reduction strategy performs the following two steps. 

\textbf{Finding nonoverlaps in bounds:} Given an element of a $\theta$-compatible $\StringProperty(\sigma_i)$ including both $\wequat{Z}{\upsilon_0 Y_0}$ and $\wequat{Z}{X_0\upsilon_1}$, $n\geq|\upsilon_0|+|\upsilon_1|-|\mathrm{overlap}(\upsilon_0,\upsilon_1)|$. In this step, the properties are traversed in the descending order wrt morphisms. If the overlap of suffix and prefix is preserved by a property $\StringProperty(\sigma_i)$, then all the properties given by $\sigma_j\prec\sigma_i$ are not considered any more.

\textbf{Alphabetic argument:} Given an element of a compatible with $\theta$ property whose factors' alphabet is of the size $k$, $n\geq k$. In this step, the properties are considered in the ascending order wrt the morphisms.

The unary bounds reduction can result in collapsing the whole string object to $\bot$, in case when the resulting lower bound in an element of $\StringProperty(\theta)$ exceeds the upper bound.

The overall reduction of a string object $\rho_{\StringObjectDom}$ is a composition of the three schemes described previously in this section.
\begin{itemize}
\item The standalone reduction of all properties in binary alphabets is applied first.
\item Then follows the universal algorithm of lower bounds cross-reduction.
\item Finally, the unary bounds reduction is applied, and the updated unary intervals are reduced.
\end{itemize}

The following lemma verifies that elements of $\StringObjectDom$ defined over the fixed set of string properties form a valid lattice wrt $\rho_{\StringObjectDom}$.

\begin{lemma}\label{Lemma:lattice}

Given abstract $\absval_1,\absval_2\in\StringObjectDom$ defined on the same set of string properties, and the reduction function $\rho_{\StringObjectDom}$:
\begin{itemize}
\item $\rho_{\StringObjectDom}\bigl(\absval_1\suprcart \absval_2\bigr) = \rho_{\StringObjectDom}\bigl(\absval_1\bigr)\suprcart\rho_{\StringObjectDom}\bigr(\absval_2\bigr)$;
\item $\rho_{\StringObjectDom}\bigl(\absval_1\infncart \absval_2\bigl)\preceq \rho_{\StringObjectDom}\bigl(\absval_1\bigr)$;
\item If $\absval_i$ are reduced, then 

\qquad $\rho_{\StringObjectDom}\biggerl\rho_{\StringObjectDom}\bigl(\absval_1\infncart \absval_2\bigr)\infncart \absval_3\biggerr = \rho_{\StringObjectDom}\biggerl\absval_1\infncart\rho_{\StringObjectDom}\bigl(\absval_2\infncart\absval_3\bigr)\biggerr$.
\end{itemize}
\end{lemma}
The first two equalities together with idempotency of $\rho_{\StringObjectDom}$ provide the absorption rules.
The commutativity of all lattice operations holds by their definition. 
The associativity of the reduced meet operation is guaranteed by the associativity of set union (when the length is updated through tracking alphabet cardinality), and by the fact that prefixes' and suffixes' lengths can only increase via the meet operation.

\subsection{Processing String Objects with Distinct Properties}
\label{Sect:distinct}

While we rely on encoding ordering of the letters in $\Sigma$, we assume that $\empt$ is less than any letter.

\begin{definition}
Let $\sigma_1$ and $\sigma_2$ be standard morphisms. Then, given $\delta\in \Sigma_{\delta_1}(\sigma_1)$, $\delta\in\Sigma_{\delta_2}(\sigma_2)$,  $\min(\sigma_1,\sigma_2)$ maps $\delta$ to $\min(\delta_1,\delta_2)$.

A supremum of morphisms $\sigma_1\supr\sigma_2$ is defined as $\min(\sigma_1,\sigma_2)^*$, where $\sigma^*$ is a fixed point of $\sigma$. 
\end{definition}

The morphism join operations satisfies the commutativity and associativity property, hence, morphisms and string properties form semilattices with finite ascending chains. We assume that the top morphism mapping everything to $\empt$ is included in all these semilattices, but do not specify it explicitly, since the corresponding property value is always trivial.

Let arbitrary $\absval_1,\absval_2\in\StringObjectDom$ be given. Given the morphism semilattices $\Lang(\StringProperty)_1$ and $\Lang(\StringProperty)_2$ in which $\absval_1$ and $\absval_2$ are defined, respectively, their combined morphism semilattice is defined as  $\{\eta\mid \eta=\sigma_i\supr\theta_j,\sigma_i\in\Lang(\StringProperty)_1\logand \theta_j\in\Lang(\StringProperty)_2\}$.

After the combined morphism semilattice is specified, $\absval_1\supr\absval_2$ and $\absval_1\infn\absval_2$ are computed as described in Sect.~\ref{Sect:reduction}. After the computation, all the properties whose values are exactly the images of properties defined with lesser morphisms are removed as redundant.

Examples of lattice operations on abstract values defined via distinct morphisms are given in Fig.~\ref{Fig:combined}.

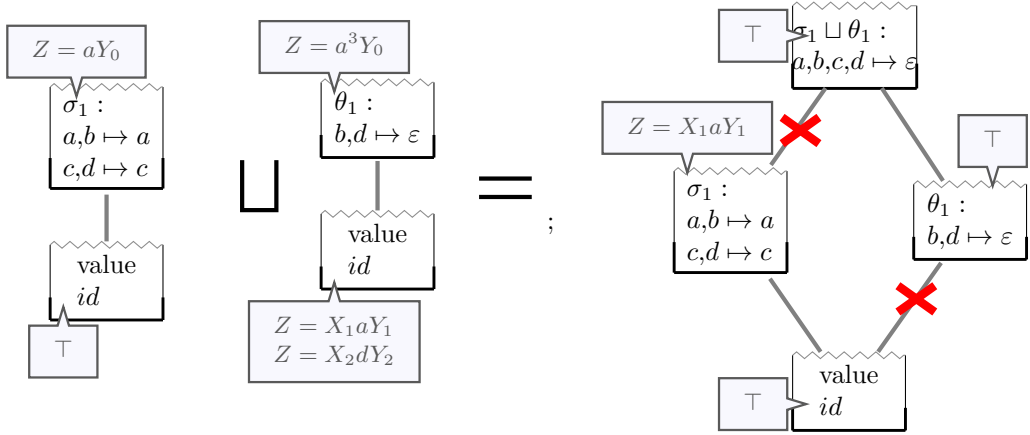
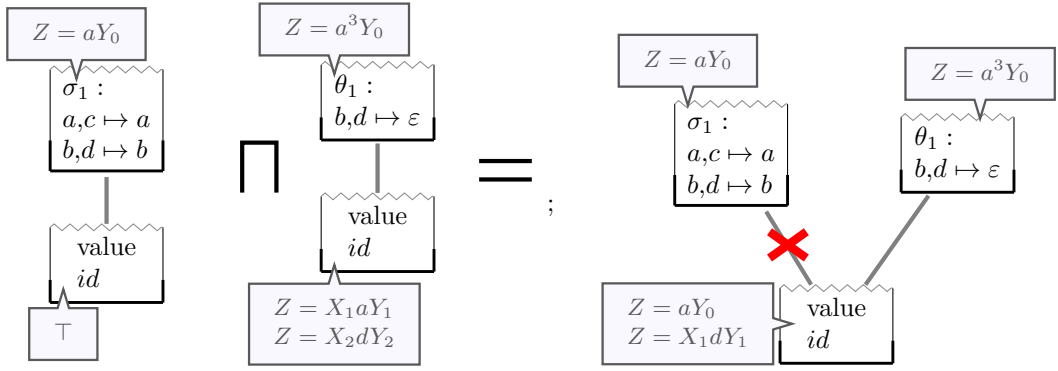
\begin{figure*}[!htb]
\centering
\begin{subfigure}[t]{\textwidth}
\begin{tabular}{ccc}
\begin{tabular}{c}
\begin{tikzpicture}[scale=0.7]

\draw (-1,-1.75) pic{eqnode ={val,$\begin{array}{l}\text{value}\\id\end{array}$,,2,1}};

\node (valval) at (-0.75,-2.35) [draw=black!30!gray,text=black!30!gray,line width=0.2ex,rectangle callout,shape border rotate=0, minimum height=3ex,inner sep = 1ex,fill=myBlue!20!white, callout relative pointer={(0,0.25)},callout pointer width=1ex] {\small $\begin{array}{l}\top\end{array}$};

\draw (-1,0.75) pic{eqnode ={s1,$\begin{array}{l}\sigma_1:\\a{,}b\mapsto a\\ c{,}d\mapsto c\end{array}$,,2,1.35}};

\node (vals1) at (-0.5,3.35) [draw=black!30!gray,text=black!30!gray,line width=0.2ex,rectangle callout,shape border rotate=0, minimum height=3ex,inner sep = 1ex,fill=myBlue!20!white, callout relative pointer={(0,-0.25)},callout pointer width=1ex] {\small $\begin{array}{l}Z=aY_0\end{array}$};

\draw[-,ultra thick,gray][-] (val) -> (s1);

\end{tikzpicture} \\ 
\begin{tikzpicture}
\node at (0,0) {};
\end{tikzpicture} \end{tabular}
&\;\;\begin{tabular}{c}
\begin{tikzpicture}[scale=0.7]

\draw[ultra thick, black] (-2.5,0.7) -- (-2.5,-0.25) -- (-1.85,-0.25) -- (-1.85,0.7);

\draw[ultra thick, black] (2,0.35) -- (3,0.35);
\draw[ultra thick, black] (2,0) -- (3,0);

\draw (-1,-1.75) pic{eqnode ={val,$\begin{array}{l}\text{value}\\id\end{array}$,,2,1}};

\node (valval) at (-0.75,-2.75) [draw=black!30!gray,text=black!30!gray,line width=0.2ex,rectangle callout,shape border rotate=0, minimum height=3ex,inner sep = 1ex,fill=myBlue!20!white, callout relative pointer={(0,0.25)},callout pointer width=1ex] {\small $\begin{array}{l}Z=X_1 a Y_1\\ Z=X_2 d Y_2\end{array}$};

\draw (-1,0.75) pic{eqnode ={s1,$\begin{array}{l}\theta_1:\\b{,}d\mapsto\empt\end{array}$,,2,1}};

\node (vals1) at (-0.75,2.85) [draw=black!30!gray,text=black!30!gray,line width=0.2ex,rectangle callout,shape border rotate=0, minimum height=3ex,inner sep = 1ex,fill=myBlue!20!white, callout relative pointer={(0,-0.25)},callout pointer width=1ex] {\small $\begin{array}{l}Z=a^3 Y_0\end{array}$};

\draw[-,ultra thick,gray][-] (val) -> (s1);

\end{tikzpicture}\\
\begin{tikzpicture}
\node at (0,0) {};
\end{tikzpicture}\end{tabular};
&\;\begin{tabular}{c}
\begin{tikzpicture}[scale=0.7]

\draw (-1,-3) pic{eqnode ={val,$\begin{array}{l}\text{value}\\id\end{array}$,,2,1}};

\node (valval) at (-1.75,-2.5) [draw=black!30!gray,text=black!30!gray,line width=0.2ex,rectangle callout,shape border rotate=0, minimum height=3ex,inner sep = 1ex,fill=myBlue!20!white, callout relative pointer={(0.25,0)},callout pointer width=1ex] {\small $\begin{array}{l}\top\end{array}$};

\draw (-3.25,0) pic{eqnode ={s1,$\begin{array}{l}\sigma_1:\\a{,}b\mapsto a \\ c{,}d\mapsto c\end{array}$,,2,1.35}};

\draw (1.3,0.25) pic{eqnode ={t1,$\begin{array}{l}\theta_1:\\b{,}d\mapsto \empt\end{array}$,,2,1}};

\draw (-1,3.5) pic{eqnode ={st1,$\begin{array}{l}\sigma_1\supr\theta_1:\\a{,}b{,}c{,}d\mapsto\empt\end{array}$,,2.2,1}};

\node (valst1) at (-1.75,4.5) [draw=black!30!gray,text=black!30!gray,line width=0.2ex,rectangle callout,shape border rotate=0, minimum height=3ex,inner sep = 1ex,fill=myBlue!20!white, callout relative pointer={(0.25,0)},callout pointer width=1ex] {\small $\begin{array}{l}\top\end{array}$};

\node (vals1) at (-3,2.75) [draw=black!30!gray,text=black!30!gray,line width=0.2ex,rectangle callout,shape border rotate=0, minimum height=3ex,inner sep = 1ex,fill=myBlue!20!white, callout relative pointer={(0,-0.25)},callout pointer width=1ex] {\small $\begin{array}{l}Z=X_1aY_1\end{array}$};

\node (valt1) at (2.8,2.5) [draw=black!30!gray,text=black!30!gray,line width=0.2ex,rectangle callout,shape border rotate=0, minimum height=3ex,inner sep = 1ex,fill=myBlue!20!white, callout relative pointer={(0,-0.25)},callout pointer width=1ex] {\small $\begin{array}{l}\top\end{array}$};

\draw[-,ultra thick,gray][-] (val) -> (s1);
\draw[-,ultra thick,gray][-] (s1) -> (st1);
\draw[-,ultra thick,gray][-] (t1) -> (st1);
\draw[-,ultra thick,gray][-] (val) -> (t1);

\draw[-,line width=0.75ex,red][-] (1,-0.25) -- (1.7,-0.75);
\draw[-,line width=0.75ex,red][-] (1,-0.75) -- (1.7,-0.25);

\draw[-,line width=0.75ex,red][-] (-0.5,2.45) -- (-1.2,2.95);
\draw[-,line width=0.75ex,red][-] (-0.5,2.95) -- (-1.2,2.45);

\end{tikzpicture}
\end{tabular}
\end{tabular}
\caption{Joining string objects. Properties $\StringProperty(\theta_1)$ and $\StringProperty(\sigma_1\supr\theta_1)$ are deleted as trivial.}
\end{subfigure}

\medskip
\begin{subfigure}[t]{\textwidth}
\begin{tabular}{ccc}
\begin{tabular}{c}
\begin{tikzpicture}[scale=0.7]

\draw (-1,-1.75) pic{eqnode ={val,$\begin{array}{l}\text{value}\\id\end{array}$,,2,1}};

\node (valval) at (-0.75,-2.35) [draw=black!30!gray,text=black!30!gray,line width=0.2ex,rectangle callout,shape border rotate=0, minimum height=3ex,inner sep = 1ex,fill=myBlue!20!white, callout relative pointer={(0,0.25)},callout pointer width=1ex] {\small $\begin{array}{l}\top\end{array}$};

\draw (-1,0.75) pic{eqnode ={s1,$\begin{array}{l}\sigma_1:\\a{,}c\mapsto a\\ b{,}d\mapsto b\end{array}$,,2,1.35}};

\node (vals1) at (-0.5,3.35) [draw=black!30!gray,text=black!30!gray,line width=0.2ex,rectangle callout,shape border rotate=0, minimum height=3ex,inner sep = 1ex,fill=myBlue!20!white, callout relative pointer={(0,-0.25)},callout pointer width=1ex] {\small $\begin{array}{l}Z=aY_0\end{array}$};

\draw[-,ultra thick,gray][-] (val) -> (s1);

\end{tikzpicture} \\ 
\begin{tikzpicture}
\node at (0,0) {};
\end{tikzpicture} \end{tabular}
&\;\;\begin{tabular}{c}
\begin{tikzpicture}[scale=0.7]

\draw[ultra thick, black] (-2.5,-0.25) -- (-2.5,0.7) -- (-1.85,0.7) -- (-1.85,-0.25);

\draw[ultra thick, black] (2,0.35) -- (3,0.35);
\draw[ultra thick, black] (2,0) -- (3,0);

\draw (-1,-1.75) pic{eqnode ={val,$\begin{array}{l}\text{value}\\id\end{array}$,,2,1}};

\node (valval) at (-0.75,-2.75) [draw=black!30!gray,text=black!30!gray,line width=0.2ex,rectangle callout,shape border rotate=0, minimum height=3ex,inner sep = 1ex,fill=myBlue!20!white, callout relative pointer={(0,0.25)},callout pointer width=1ex] {\small $\begin{array}{l}Z=X_1 a Y_1\\ Z=X_2 d Y_2\end{array}$};

\draw (-1,0.75) pic{eqnode ={s1,$\begin{array}{l}\theta_1:\\b{,}d\mapsto\empt\end{array}$,,2,1}};

\node (vals1) at (-0.75,2.85) [draw=black!30!gray,text=black!30!gray,line width=0.2ex,rectangle callout,shape border rotate=0, minimum height=3ex,inner sep = 1ex,fill=myBlue!20!white, callout relative pointer={(0,-0.25)},callout pointer width=1ex] {\small $\begin{array}{l}Z=a^3 Y_0\end{array}$};

\draw[-,ultra thick,gray][-] (val) -> (s1);

\end{tikzpicture}\\
\begin{tikzpicture}
\node at (0,0) {};
\end{tikzpicture}\end{tabular};
&\;\begin{tabular}{c}
\begin{tikzpicture}[scale=0.7]

\draw (-1,-3) pic{eqnode ={val,$\begin{array}{l}\text{value}\\id\end{array}$,,2,1}};

\node (valval) at (-2.75,-2.25) [draw=black!30!gray,text=black!30!gray,line width=0.2ex,rectangle callout,shape border rotate=0, minimum height=3ex,inner sep = 1ex,fill=myBlue!20!white, callout relative pointer={(0.25,0)},callout pointer width=1ex] {\small $\begin{array}{l}Z=a Y_0 \\ Z=X_1 d Y_1\end{array}$};

\draw (-3,0) pic{eqnode ={s1,$\begin{array}{l}\sigma_1:\\a{,}c\mapsto a \\ b{,}d\mapsto b\end{array}$,,2,1.35}};

\draw (1.3,0.25) pic{eqnode ={t1,$\begin{array}{l}\theta_1:\\b{,}d\mapsto \empt\end{array}$,,2,1}};

\node (vals1) at (-2.75,2.75) [draw=black!30!gray,text=black!30!gray,line width=0.2ex,rectangle callout,shape border rotate=0, minimum height=3ex,inner sep = 1ex,fill=myBlue!20!white, callout relative pointer={(0,-0.25)},callout pointer width=1ex] {\small $\begin{array}{l}Z=aY_0\end{array}$};

\node (valt1) at (2.8,2.5) [draw=black!30!gray,text=black!30!gray,line width=0.2ex,rectangle callout,shape border rotate=0, minimum height=3ex,inner sep = 1ex,fill=myBlue!20!white, callout relative pointer={(0,-0.25)},callout pointer width=1ex] {\small $\begin{array}{l}Z=a^3 Y_0\end{array}$};

\draw[-,ultra thick,gray][-] (val) -> (s1);
\draw[-,ultra thick,gray][-] (val) -> (t1);

\draw[-,line width=0.75ex,red][-] (-0.5,-0.5) -- (-1.2,-1);
\draw[-,line width=0.75ex,red][-] (-0.5,-1) -- (-1.2,-0.5);

\end{tikzpicture}
\end{tabular}
\end{tabular}
\caption{Meeting string objects. The equation $Z=a Y_0$ is added to the value via propagation. The property $\StringProperty(\sigma_1)$ is deleted, being derived from the updated value.}
\end{subfigure}
\caption{Example of object operations with distinct morphisms semilattices}
\label{Fig:combined}
\end{figure*}

\section{Abstracting String Operations}\label{Sect:operations}

We are focused on analysing computations making use of typical string manipulating functions, such as concatenation, checking inclusion, string replacements, and so on. Most of these functions can be expressed in terms of others, so we choose the following function basis:
\begin{itemize}
\item string concatenation $\oconcat{s_1}{s_2}$,
\item finding an index of the first occurrence of string $s_2$ in string $s_1$, $\oindex{s_1}{s_2}$,
\item taking a suffix of a string $s_1$ from a given position $\osubstr{s_1}{n}$,
\item and the replacement method $\oreplace{s_1}{s_2}{s_3}$, which takes out the first occurrence of $s_2$ in $s_1$ and inserts $s_3$ instead.
\end{itemize}

\begin{figure}[!htb]
There $M$ is an interpretation, mapping identifiers to concrete values. For the sake of brevity, $\oxbr{s_1}_M = \omega_1$, $\oxbr{s_2}_M = \omega_2$, $\oxbr{s_3}_M = \omega_3$, $\oxbr{n}_M = N$.

$$\footnotesize
\begin{array}{rcl}
\oxbr{\oconcat{s_1}{s_2}}_{M}&=&\omega_1\omega_2 \\ 
\oxbr{\osubstr{s_1}{n}}_{M}&=&\begin{cases} \omega_1,\;N\leq 0\\
\upsilon_2,\;\omega_1 = \upsilon_1 \upsilon_2,\;|\upsilon_1|=N,\;0\leq N\leq |\omega_1|-1\\
\empt,\;N\geq|\omega_1|\end{cases} \\ 
\oxbr{\oindex{s_1}{s_2}}_M&=& \begin{cases}|\upsilon_1|,\;\omega_1 = \upsilon_1\omega_2\upsilon_2\\-1,\;\text{otherwise}\end{cases}\\
\oxbr{\oreplace{s_1}{s_2}{s_3}}_M&=& \begin{cases}omega_1,\;omega_1 \notin\Sol_Z(\wequat{Z}{X \omega_2 Y})\\
\omega_3\omega_1,\;\omega_2=\empt\\ \upsilon_1\omega_3\upsilon_2,\;\omega_1 = \upsilon_1\omega_2\upsilon_2,\omega_2 = \omega'_2\delta, \upsilon_1\omega'_2\notin\Sol(\wequat{Z}{X\omega_2 Y})\end{cases}\\
\oxbr{\ochar{s_1}{n}}_M&=& \begin{cases}\gamma,\;\omega_1 = \upsilon_1\gamma\upsilon_2, |\upsilon_1|=N,\; 0\leq N\leq |\omega_1|-1\\\empt,\;\text{otherwise}\end{cases}
\end{array}
\vspace*{-1ex}
$$
\caption{Concrete semantics of basic operations in string domain}
\label{Fig::concsem}
\end{figure}

The concrete semantics of the operations is presented in Fig.~\ref{Fig::concsem}. We also add operation $\mathtt{charAt}$ there, since it is used in our examples, although $\mathtt{charAt}$ can be considered as a composition of the basic operations. Other operations are omitted for the sake of brevity.

The string functions are mainly monotone with respect to morphisms. Namely, given any concrete $\omega_1$, $\omega_2$, and $\omega_3$, and $n\in\mathbb{N}$:
\begin{itemize}
\item $\sigma(\oconcat{\omega_1}{\omega_2})=\oconcat{\sigma(\omega_1)}{\sigma(\omega_2)}$;
\item if $\Sigma_{\empt}(\sigma)=\varnothing$, then $\sigma(\osubstr{\omega_1}{n}) = \osubstr{\sigma(\omega_1)}{n}$. In the case of erasing morphisms, $\sigma(\osubstr{\omega_1}{n})$ ends with $\osubstr{\sigma(\omega_1)}{n}$;
\item If $\oindex{\sigma(\omega_1)}{\sigma(\omega_2)}<0$, that is, the morphic image of $\omega_2$ does not occur in the morphic image of $\omega_1$ as a substring, then $\oindex{\omega_1}{\omega_2}$ is also negative. If $\oindex{\omega_1}{\omega_2}\geq 0$, then $\biggerl\oindex{\sigma(\omega_1)}{\sigma(\omega_2)}\leq \oindex{\omega_1}{\omega_2}\biggerr$. 
\end{itemize}
The replacement method is non-monotone. Still, it preserves a morphic image of a suffix unaffected with the replacement.

The mentioned feature of the string operations can be used to make string properties cross-reduction lightweight. Indeed, being forced to perform the complete reduction procedure after any operation on $\StringObjectDom$ domain looks like a significant complexity overhead. Given arguments that are already reduced, most of the operations traversing the whole set of lower bounds require neither  an additional propagation upwards, nor a downward propagation of preserved factors. 

\begin{example}\label{Example::concat}
Given two non-empty elements $\langle\bot, \infopat{\absval_1}\rangle$, $\langle\bot, \infopat{\absval_2}\rangle$ of a string property lattice $\Bool\tcartesian(\Point\treduced(\Prefix\tcartesian\Suffix)\treduced\FactorCode)$, their concatenation is defined as follows\footnote{The non-empty lower bounds are chosen there for simplicity: when the lower bounds can be concretised to $\empt$, the concatenation algorithm becomes more involved.}. 
\begin{itemize}
\item  if both elements of $\Point$ domain are non-top, concatenate them and return the corresponding constant equation; 
\item otherwise, take the prefix element from $\infopat{\absval_1}$ making the prefix from $\infopat{\absval_2}$ a factor, the suffix element from $\infopat{\absval_2}$, making the $\infopat{\absval_1}$ suffix a factor, and merge the factor codes.
\end{itemize}
All these steps can be done uniformly for all lower bounds, and all the preserved subwords in them are retained in the result of the operation without any additional reduction.
\end{example}

Example~\ref{Example::concat} shows that the concatenation is computed naturally on abstract string objects. 

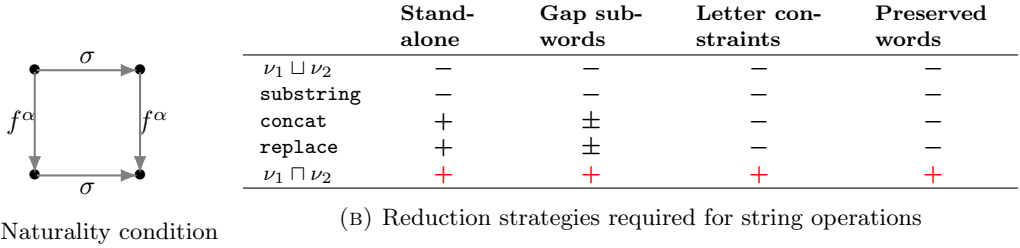
\begin{figure}[htb]
\begin{subfigure}[c]{0.27\textwidth}
\centering
\bigskip
\bigskip
\medskip
\begin{tikzpicture}[scale=0.7]
\node[inner sep=-0.75ex] (val) at (0,2) {$\bullet$};
\node[inner sep=-0.75ex] (h1) at (0,0) {$\bullet$};
\node[inner sep=-0.75ex] (s1) at (2,2) {$\bullet$};
\node[inner sep=-0.75ex] (sh) at (2,0) {$\bullet$};
 
\node (1) at (-0.25,1) {$f^{\alpha}$};
\node (2) at (1,-0.25) {$\sigma$};
\node (3) at (2.25,1) {$f^{\alpha}$};
\node (2) at (1,2.25) {$\sigma$};
 
\draw[-{Latex},thick,gray] (val) -> (h1);
\draw[-{Latex},thick,gray] (val) -> (s1);
\draw[-{Latex},thick,gray] (s1) -> (sh);
\draw[-{Latex},thick,gray] (h1) -> (sh);
 
\end{tikzpicture}
\caption{Naturality condition}
\end{subfigure}\hfill
\begin{subfigure}[c]{0.73\textwidth}
\centering\footnotesize
\begin{tabular}{lcccc}\bfseries
 &\bfseries\begin{tabular}{l}\footnotesize{Stand-}\,\;\\{alone}\end{tabular} & \bfseries\begin{tabular}{l}{Gap sub-}\\words\,\end{tabular} & \bfseries\begin{tabular}{l}{Letter con-}\\{straints}\,\end{tabular} & \bfseries\begin{tabular}{l}{Preserved}\,\\words\end{tabular}
\\\hline
$\,\,\nu_1 \supr \nu_2$ & $\fancym$ & $\fancym$ & $\fancym$ & $\fancym$\\
$\,\mathtt{substring}$ & $\fancym$ & $\fancym$ & $\fancym$ & $\fancym$\\
$\,\mathtt{concat}$& $\fancyp$ & $\fancypm$ & $\fancym$ & $\fancym$\\
$\,\mathtt{replace}$ & $\fancyp$ & $\fancypm$ & $\fancym$ & $\fancym$\\
$\,\,\nu_1 \infn \nu_2$ & $\textcolor{red}{\fancyp}$ & $\textcolor{red}{\fancyp}$ & $\textcolor{red}{\fancyp}$ & $\textcolor{red}{\fancyp}$
\\\hline\end{tabular}
\caption{Reduction strategies required for string operations}
\end{subfigure}
\caption{Naturality of abstract string operations}
\label{Diag::adequacy}
\end{figure}

\begin{definition}\label{def::adequacy}
Given an abstract string operation $f^{\alpha}$, we call it natural if, given any standard $\sigma_1$ and $\sigma_2$, $\sigma_2 \circ f^{\alpha} = f^{\alpha}\circ\sigma_2$ on any string lower bound defined by factors preserved by $\sigma_2$ wrt $\sigma_1$. 
\end{definition}

If a string operation is natural, i.e. the diagram given in Fig.~\ref{Diag::adequacy}, left part, commutes on the properties preserved by $\sigma_2$, then we are guaranteed that the upward propagation and the preserved factor propagation in the cross-reduction are already done when computing $f^{\alpha}$ on reduced arguments. Note that the following rule also holds by default in the case of the natural string operations: \textit{given the longest string in the factor code of a property, the string object length cannot be less than its length}. 

In the table in Fig.~\ref{Diag::adequacy}, right, we show whether basic string operations require a certain sort of explicit reduction. The sign $\fancypm$ marking the gap subwords propagation in concatenation and replacement operations points out that the subwords can be propagated from certain positions only, i.e. from the bounds where the concatenation occurs. Most of the abstract string operations are therefore lightweight in terms of reduction; the only exception is the meet operation, which is highly unnatural. In the cost of its computational complexity, the meet operation can highly improve preciseness of the abstract analysis, when it is used to create \textit{assume}, i.e. context-based, constraints on string objects~\cite{pagai}. 

\begin{figure}[!htb]
\begin{subfigure}[t]{\textwidth}
\centering
\begin{tabular}{rllr}
... & $\mathtt{/\hspace*{-0.8ex}* x = }\,\top{ */}$ && $x\mapsto\top$\\
1 & \multicolumn{2}{l}{$\mathtt{y = x\,?\;}\quote\hspace*{-0.5ex}\mathtt{<\hspace{-0.65ex}tag\hspace{-0.65ex}>}\hspace*{0.1ex}\quote\mathtt{ +\,x \,:\;}\quote\quote\mathtt{;}$} & $\objval{y}\mapsto\{\wequat{Z}{\hspace*{-0.5ex}\mathtt{<\hspace{-0.65ex}tag\hspace{-0.65ex}>}Y}\},$\\
 &&& $\objlen{y}\mapsto[0]\cup[5;+\infty)$\\
2 & \textbf{let }$\mathtt{z = \quote\hspace*{-0.3ex}?\hspace*{0.3ex}\quote}\mathtt{;}$ &\multicolumn{2}{r}{\qquad $\objval{z}\mapsto \mathbb{?},\objlen{z}\mapsto 1$}\\
3 & \textbf{if }$\mathtt{(y)}$ &\multicolumn{2}{r}{\qquad\qquad\qquad\qquad\qquad(Yields guard condition $\objlen{y}>0$)}\\
4 & \quad $\mathtt{z = y.charAt(4)}\mathtt{;}$ & \multicolumn{2}{r}{$\objlen{y}\mapsto [5;+\infty),\objlen{z}\mapsto 1$}\\
5 & \multicolumn{2}{l}{\textbf{if }$\mathtt{(!z)}$}& (never holds)\\
6 & \multicolumn{2}{l}{\quad\textbf{return }Error;}& (is unreachable)\\
7 & \multicolumn{2}{l}{\textbf{else}{\quad\textbf{return }\texttt{z};}}& ($\objval{z}=\top, \objlen{z}=1$)\\\end{tabular}
\caption{Non-emptiness length condition}
\end{subfigure}

\begin{subfigure}[t]{\textwidth}
\medskip
\centering
\begin{tabular}{rllr}
... & $\mathtt{/\hspace*{-0.6ex}* x = }\,\quote a\quote, \mathtt{y = }\,\quote b\quote{ */}$ && $x\mapsto\mathbf{a},y\mapsto\mathbf{b}$\\
1 & \multicolumn{2}{l}{\textbf{let} $\mathtt{z = x+y;}$} & $z\mapsto\mathbf{ab}$\\
2 & \textbf{if }$(\top)$ &\\
\multicolumn{4}{l}{$\,\mathtt{/\hspace*{-0.2ex}* }$ $\mathbf{ab}\in\concr{z}$, having no free $\mathbf{a}$ occurrence $\mathtt{*/}$}\\
3 & \quad $\mathtt{z = x + \top + z;}$ &\multicolumn{2}{r}{\;$\objval{z}\mapsto\{Z=\mathbf{a}Y, Z=X\mathbf{ab}\},\objlen{z}\mapsto[2;+\infty)$}\\
\multicolumn{4}{l}{$\,\mathtt{/\hspace*{-0.2ex}* }$ any element of $\concr{z}$ includes an $\mathbf{a}$ occurrence not inside $\mathbf{ab}$ $\mathtt{*/}$}\\
4 & $\mathtt{z = x + z;}$ &\multicolumn{2}{r}{\qquad\quad\quad\;$\objval{z}\mapsto\{Z=\mathbf{aa}Y, Z=X\mathbf{ab}\},\objlen{z}\mapsto[3;+\infty)$}\\
5 & \multicolumn{2}{l}{\textbf{while }$\mathtt{(\oindex{z}{x+y}\geq 0)}$}& (holds infinitely)\\
6 & \multicolumn{2}{l}{\quad$\mathtt{z=\oreplace{z}{x+y}{x+\quote\_\quote + y};}$}
& $\objval{z}\mapsto\{Z=\mathbf{a\_}Y{,}Z=X\mathbf{ab}\}, $\\
&&\multicolumn{2}{r}{$\objlen{z}\mapsto[+\infty]$}\\
\end{tabular}
\caption{Proving that the silly sanitizer loops infinitely}
\end{subfigure}

\begin{subfigure}[t]{\textwidth}
\medskip
\centering
\begin{tabular}{rllr}
... & $\mathtt{/\hspace*{-0.6ex}* x = }\,\quote\hspace*{-0.3ex}fstTag\quote, \mathtt{y = }\,\quote\hspace*{-0.6ex}secondTag\quote{ */}$ &\multicolumn{2}{r}{$x\mapsto\mathbf{fstTag},y\mapsto\mathbf{secondTag}$}\\
1 & \multicolumn{2}{l}{\textbf{let} $\mathtt{z = \quote\hspace*{-0.75ex}<\quote+x+\quote\hspace*{-0.75ex}>\quote+\top+\quote\hspace*{-0.75ex}</\quote + x +\quote\hspace*{-0.75ex}>\quote;}$} & $\StringProperty_{\sigma}(z)\mapsto\{\wequat{Z}{\mathbf{<>}Y},\wequat{Z}{X\mathbf{<>}}\}$\\
&\multicolumn{3}{r}{$\objval{z}\mapsto\{\wequat{Z}{\mathbf{<fstTag>}Y},\wequat{Z}{X\mathbf{</fstTag>}}\},\objlen{z}\mapsto[17;+\infty)$}\\
2 & \textbf{while }$(\top)$ & \\
3 & \multicolumn{2}{l}{\quad $\mathtt{z = \quote\hspace*{-0.75ex}< \quote + y + \quote\hspace*{-0.75ex}>\quote + z;}$} &$\StringProperty_{\sigma}(z)\mapsto\{\wequat{Z}{\mathbf{<>}Y},\wequat{Z}{X\mathbf{<>}}\}$\\
&\multicolumn{3}{r}{$\objval{z}\mapsto\{\wequat{Z}{\mathbf{<}Y},\wequat{Z}{X\mathbf{</fstTag>}}\},\objlen{z}\mapsto [17;+\infty)$}\\\
4 & \multicolumn{2}{l}{$\mathtt{w = \quote\hspace*{-0.75ex}<\quote + (\top\;? \;x\;:\;y) + \quote\hspace*{-0.75ex}>\quote;}$} & $\StringProperty_{\sigma}(w)\mapsto \mathbf{<>},\objlen{w}\mapsto[8;12)$\\
&\multicolumn{3}{r}{$\objval{w}\mapsto\{\wequat{Z}{\mathbf{<}Y},\wequat{Z}{X\,\mathbf{Tag>}}\}$}\\
5 & $\mathtt{z = \osubstr{z}{1}}$& \multicolumn{2}{r}{$\StringProperty_{\sigma}(z)\mapsto\{\wequat{Z}{\mathbf{>}Y},\wequat{Z}{X\mathbf{>}}\}$}\\
&\multicolumn{3}{r}{$\objval{z}\mapsto\{\wequat{Z}{X\mathbf{</fstTag>}}\},\objlen{z}\mapsto [16;+\infty)$}\\\
6 & {\textbf{if }$(\mathtt{\oindex{z}{w} == 0)}$ }& \multicolumn{2}{r}{(never holds: $\StringProperty_{\sigma}(w)$ mismatches with\;}\\
&\multicolumn{3}{r}{ a prefix in $\StringProperty_{\sigma}(z)$)}\\
7 & \multicolumn{2}{l}{\quad\textbf{return }Error;}& (is unreachable)\\
\end{tabular}
\caption{Utilizing string property $\sigma(\delta)=\begin{cases}\delta, \delta\in\{<,>\}\\ \empt,\text{ otherwise}\end{cases}$}
\end{subfigure}
\caption{Some program invariants discovered by the abstract interpretation over the string object lattice. Constant string are given in bold inside the abstract values, constant string properties $\nu\mapsto\{\wequat{Z}{\omega}\}$ are shortcut to $\nu\mapsto\omega$.}
\label{Fig::Prog}\end{figure}

Now let us consider three small \JS{} programs given in Fig.~\ref{Fig::Prog}, and the abstract values computed along their traces. 

\medskip
\paragraph{\textit{Example }(\textsc{A}): split length interval meets with guard conditions.} In the line 1, a ternary conditional expression checks whether $\mathtt{x}$ is empty. If it is not empty, $\mathtt{x}$ value is supplied with a proper tag. Then the line 4 looks at 5-th character of the string, and if it is empty, an error is returned. Since string objects process string lengths like other properties, they take into account the empty string possibility separately, and the split interval $[0]\cup[5;+\infty)$ is constructed for capturing $\mathtt{x}$ values' possible length.   

Another interesting detail in this abstract trace is the guard condition propagated to the line 4 from the line 3. In the line 3, $\mathtt{y}$ is tested for non-emptiness. Hence, in the line 4, the length of $\mathtt{y}$ cannot be equal to $0$ any more.

\medskip
\paragraph{\textit{Example }(\textsc{B}): silly sanitizer looping forever.} The loop in lines 5-6, while seemingly aims at splitting all the words $ab$ by $\_$, always goes into infinite loop, which is proved by the abstract interpretation. Indeed, the prefix and the suffix of the abstract object value given in the line 6 are non-overlapping, hence the string $ab$ is guaranteed to be preserved in all the words in $\concr{z}$. The abstract interpreter sees that the guard condition on loop termination, namely $\oindex{z}{x+y}<0$, never holds, and returns $z\mapsto\bot$ as the result of the analysis.

\medskip
\paragraph{\textit{Example }(\textsc{C}): letters preserved by a string property can boost analysis preciseness.} While in the line 1, the element of $\StringProperty_{\sigma}$ contains no constraint on $\concr{z}$ not imposed by $\objval{z}$, in the line 2 it contains a unique constraint not derived from $\objval{z}$ any more: that $\sigma(z)$ still starts with $<>$. It may seem that $\StringProperty_{\sigma}(w)$ in the line duplicates an image of $w$'s value, but actually it states a bit more --- namely, that for any $\omega\in\concr{w}$, $|\omega|_< = |\omega|_> = 1$. Both letters $<$, $>$ are preserved by $\sigma$. Hence, the substring method call in the line 5 takes out only the first letter from the prefix equation\footnote{The method call also deletes the first letter from the suffix equation in the property, because the string $<>$ is in $\concr{\StringProperty_{\sigma}(z)}$.} in $\StringProperty_{\sigma}(z)$. Now the condition in the line 6 never holds due to the monotonicity of $\mathtt{indexOf}$ method via morphic images.

\section{Related Works, Discussion and Conclusion}\label{Sect:last}

The construction given in this paper seems to be a first attempt to combine  string properties expressed by means of word equations and morphisms in a reduced product being an abstract string domain. To our knowledge, cross-reduction procedures for string properties are not yet widely adopted in abstract interpretation. Existing frameworks (e.g. LiSA~\cite{TwinAut}), while consider multiple domains, primarily utilize Cartesian products, or the simplest reduced product normalizing a product with a bottom element to $\bot$. Similarly, the approach described in~\cite{YuBook}, while capable of simultaneously tracking string length and regular language membership, also relies on a Cartesian product.

String domains for dense languages are well-established~\cite{CostantiniFerrara,Mastroieni}: these include prefix-suffix domains, and domains counting specific letters (that can be expressed in the string objects by properties mapping all the letters except one into $\empt$). Finite automata are also extensively used. However, a known challenge is that regular languages can form infinite ascending chains, and require widening~\cite{wideningAut}. Within the framework in the book~\cite{YuBook}, processing words with prefixes in dense languages thus results in over-generalization of abstract values. The Tarsis automata framework addresses this issue through introduction of $\top$-marked transitions in automata~\cite{TwinAut}. Still, the lengths of the strings are over-generalized or even lost in this case. 

Papers~\cite{MidtggardLatticeValRegular,LatticeAutomata} propose an elegant approach to the abstract interpretation: the authors build their frameworks of lattice regexes and lattice automata over arbitrary atomistic lattices. Actually, the Tarsis lattice based on automata with $\top$-valued transitions can be considered as an advanced practical application of their idea.

The combination of abstract domains over a common concrete domain to improve precision is formalized by the reduced product in the seminal work~\cite{CousotReduced}. Specific applications of this idea include using a reduced product of length and buffer size domains to verify memory safety in C string manipulations~\cite{ModularCString}. Our work is partly inspired by the general framework for language-specific reduced products and adopts the associated notation from the paper~\cite{MultilingualAbstract}. Further relevant concepts include the delayed product for dynamically trading precision for efficiency~\cite{DelayedProduct}. The use of linear transformations to capture program properties in numeric domains~\cite{RotatingBoxes,ComposingBoxes} is a direct analogue to our method introducing custom properties via string morphisms. 

Outside the abstract interpretation scope, constraint solvers extensively use combined analyses of string values and lengths~\cite{Noodler,Ostritch,Reynolds,z3,BlackOstritch,lastOstritch}. For the regular constraints, corresponding lasso automata provide lengths estimations, and word equations are mapped into linear integer arithmetic language, in order not only to track individual lengths of the analysed string parameters, but also to capture relations between them. Some solvers use known upper bounds on word equations solutions lengths to restrict the search space as well~\cite{woorpje}. A framework using string constraints tracking their morphic images in commutative monoids is presented in the paper~\cite{RummerParikh}. The similar part of our framework, the one processing unary string properties, is still somewhat ad-hoc and underdeveloped, as compared to the cross-reduction of non-commutative morphic images.

Despite the impressive progress made in the string solving, the approach presented in this paper can give some insights on improving string analysis even in advanced cases. For example, the simple reduction Lemma~\ref{Lemma:ReducedLength} captures string properties that cannot be proved in the state-of-art solvers cvc5~\cite{Reynolds} and z3~\cite{z3}. The cross-reduction algorithms described in Sect.~\ref{Sect:reduction} can help pruning some search branches, if the morphisms determining the string properties are appropriately chosen. Since the algorithms also apply to the cases when the string objects possess distinct sets of the properties, the morphisms can even be tuned dynamically.

There are many open problems and work-in-progress on the way developing the suggested approach. First, while the superstring problem is known to be NP-complete, there exist efficient algorithms for estimating practically reasonable lower bounds on the superstring length~\cite{superlength}. Constructing join operation for anti-dictionaries (disequalities sets) is also a problem to be considered in the future work. Hence, reduction strategies involving anti-dictionaries are to be studied further. 

Finding most fitting morphisms in order to capture string properties, e.g., for translating strings by $\mathtt{toNum}$ method to numeric values in distinct notations in a complete manner~\cite{MastroieniCompleteness}, as well as strategies extracting custom string properties from programs automatically, are also fruitful future work directions. Finally, it is interesting to extend the set of string properties' equations outside the boundaries of regular languages. Inverse morphic images of word equations solutions are still unable to express some regular languages (see Lemma~\ref{lemma:expressibility} in Appendix), yet tracking the images via length-decreasing morphims (i.e. allowing letter counting) yields an undecidable theory~\cite{Durnev}. A language class of equations solutions on the images of length-preserving morphisms seems a fair trade: these languages can express equations on char-classes, but are very likely decidable.

\section{Acknowledgements}

The authors thank Egor Kichin and his research group for experimental validation of the presented approach on real projects, and Andrey Nemytykh for inspiration for developing the word-equations-based techniques of program analysis.

\bibliographystyle{splncs04}
\bibliography{lattice}

\clearpage
\section*{Appendix}

\subsection{Dense Solutions of Word Equations}\label{subsect::dense}
The following theorem gives a general characterisation of the dense languages given by word equations. In this subsection, words in mixed alphabet $\Sigma\cup\VarSet$ are called patterns. Given a pattern $\Phi(X_1,\dots, X_n)$ in $(\Sigma\cup\VarSet)^+$, a pattern language is $\Sol_{Z}(\wequat{Z}{\Phi(X_1,\dots, X_n)})$.

\begin{theorem}
(paper~\cite{Plandowski}, Theorem 16) Any word equation solution wrt a single variable either includes a pattern language or is thin.
\end{theorem}

Hence, the dense solutions to word equations can be somehow expressed in terms of patterns. However, finding an appropriate pattern ``basis'' (i.e. a finite set of patterns) for constructing exhaustive description of the solutions set is non-trivial. Let us show that sometimes such a solution set includes a union of infinite set of mutually distinct pattern languages.

A pattern $\Pat_1$ is an instance of a pattern $\Pat_2$, if $\forall \omega\bigl(\omega\in\Lang(\Pat_1)\logimpl\omega\in\Lang(\Pat_2)\bigr)$.

\begin{example}\label{example:patseries}
The solution set of the equation $\Eq: YaYX\weql XYaY$ w.r.t the variable $X$ contains an infinite set of mutually distinct pattern languages.

Indeed, all the patterns $(YaY)^k$ describe $X$-solutions of the equation $\Eq$. But, given $k_1$, $k_2$ s.t. $k_1$ is prime and $k_2\neq k_1$, $k_2\neq 1$, the pattern $(YaY)^{k_1}$ is not an instance of $(YaY)^{k_2}$. Really, the substitution $\sigma: Y\mapsto b$ applied to $(YaY)^{k_1}$ results in a word with $2\cdot k_1$ occurrences of $b$ that are to be divided equally between $2\cdot k_2$ occurrences of $Y$ in the pattern $(YaY)^{k_2}$. But that is impossible. 

Now let us assume that $(YaY)^{k_1}$, where $k_1$ is prime, is an instance of $YaY$. Then $(bab)^{k_1}$ is in $\Lang(YaY)$. Hence, the value substituted to $Y$ starts both with $ba$ (wrt the first instance of $Y$) and $bb$ (wrt the second instance of $Y$), which leads to a contradiction. 
\end{example}

Given an equation $E:$ $\wequat{\Phi(X_1,\dots, X_n)}{\Psi(X_1,\dots, X_n)}$ and substitution $\sigma:X_i\mapsto\omega_i$, we say that equation $\wequat{\Phi(X_1,\dots, X_n)}{\Psi(X_1,\dots, X_n)}\sigma$ results from $E$ by means of primitive specialization iff $\omega_i$ is strongly primitive, id est, cannot be represented as $\upsilon_1 \upsilon_2\upsilon_1$, where $|\upsilon_1|>0$. For example, the equation $\wequat{XZY}{YZX}$ can be primitively specialized  to the equation $\wequat{XZab}{abZX}$ by means of substitution $Y\mapsto ab$.


\begin{lemma}\label{Lemma::infproj}
Each dense non-trivial solution projection of 3-vars equations 1--33 (excluding the equation 7 and equations 29--33 depending on 4 or more variables) given in the paper~\cite{MakaninZoo1}, specialized by a strongly primitive $\omega$, either is a language described with the pattern $\omega X$ or $X\omega$, or includes an infinite union of pattern languages being not instances of each other.

\begin{proof}

All the solution projections of 3-vars equations from the 1--33-list in paper~\cite{MakaninZoo1} specialized by strongly primitive words can be described by series of 1-var patterns $\bigl(\Phi_1(\omega,X)\bigr)^n \Phi_2(\omega,X) $, where $\Phi_1$ and $\Phi_2$ are known patterns being words in the regular language $(X\,|\,\omega)^+$, and $\Phi_1$, $\Phi_2$ both contain at least one occurrence of the variable $X$ and $\Phi_1$ contains at least one occurrence of the word $\omega$.

The list of basic equations considered is given in Table~\ref{Table::3eq}, together with descriptions of projections of their specialized version. 

Now, similarly to the reasoning in Example~\ref{example:patseries}, we consider the set of patterns $\Pat_{n_1}=\bigl(\Phi_1(\omega,X)\bigr)^{n_1}\Phi_2(\omega,X)$,\dots, $\Pat_{n_m}=(\Phi_1(\omega,X))^{n_m}\Phi_2(\omega,X)$, \dots, where $n_i$ are prime numbers, and the substitution $\sigma: X\mapsto b$, where $b$ does not occur in $\omega$. First, we can note that if $k\neq 0$ and $k\neq n_i$, then the pattern $\Pat_{n_i}$ can never be an instance of $\bigl(\Phi_1(\omega,X)\bigr)^k \Phi_2(\omega,X) $, since the number of letters $b$ in the $(\Phi_1(\omega,b))^{n_i}$ part of $\Pat_{n_i}\sigma$ can not be equally distributed among $|\Phi_1(\omega,X)|_X\cdot k$ occurrences of $X$ variables.

Hence, the series $\Pat_{n_i}$ define a union of infinite languages being not instances of each other, unless all of them except the finite set are not instances of the pattern $\Phi_2(\omega,X)$ representing non-periodic part of all of the given patterns. Let us consider all possible forms of this non-periodic part.

\begin{itemize}
\item If $\Phi_2(\omega,X)=X$, or $\Phi_2(\omega,X)= X\omega$, then all the patterns $\bigl(\Phi_1(\omega,X)\bigr)^{m}\Phi_2(\omega,X)$ are instances of $\Phi_2(\omega,X)$. Hence, the solution projection language defined by the equation is either trivial or defined by the pattern $X\omega$. 
\item If $\Phi_2(\omega,X)=\omega X$ and $\Phi_1(\omega,X)$ starts with $\omega$, then again all the patterns of the form $\biggerl\Phi_1(\omega,X)\biggerr^{m}\Phi_2(\omega,X)$ are instances of $\Phi_2(\omega,X)$. Hence, the solution projection language defined by the equation is either trivial or defined by the pattern $\omega X$. 

If $\Phi_2(\omega,X)$ starts with $\omega$, while $\Phi_1(\omega,X)$ starts with $X$, then neither of the patterns $\biggerl\Phi_1(\omega,X)\biggerr^{n_i}\Phi_2(\omega,X)$ ($m>0$) is an instance of $\Phi(\omega,X)$, because $\Pat_{n_i}\sigma$ starts with $b$, and $\omega$ does not contain $b$.
\item (Equation 6) Given $\Phi_2 = X^2$, $\Phi_1 = X^2 \omega$, if $(\Phi_1^{n_i}\Phi_2)\sigma$ is an instance of $\Phi_2$, then $b$ starts $\omega$, which is contradictory.
\item (Equation 10) Given $\Phi_2 = X\omega^2 X$, $\Phi_1 = X\omega^2 X \omega$, if $(\Phi_1^{n_i}\Phi_2)\sigma$ is an instance of $\Phi_2$, then $\omega b = b \omega$ which is again contradictory.
\item (Equation 11) Given $\Phi_2 = (\omega X)^2$, $\Phi_1 = (\omega X)^2 \omega$, if $(\Phi_1^{n_i}\Phi_2)\sigma$ is an instance of $\Phi_2$, then again $b\omega = \omega b$, which is not possible.
\item (Equation 12, $X_1$-projection, Equation 25, $X_1$-projection) The series $(X\omega X)^n$ is already considered in Example~\ref{example:patseries}.
\item (Equations 8, 14, 24, and Equations 12, 25, $X_3$-projections) Given $\Phi_2 = X\omega^2 X$, $\Phi_1 = X\omega^2$, for all odd $n$, $\bigl((\Phi_1)^n\Phi_2\bigr)\sigma$ cannot be an instance of $\Phi_2$, since the letters $b$ cannot be arranged between the two pattern variables equally.
\end{itemize}  
   
\end{proof}
\end{lemma}

The list of 1--28 basis equations from the paper~\cite{MakaninZoo1} and used in Proposition~\ref{Lemma::infproj} depending on 2 or 3 variables is given below. Its projections after the variable specialization are given in terms of pattern languages if possible. $X$ means a trivial pattern language, ``thin'' stands for the thin projections, ``inf'' stands for the infinite union of pattern languages.

We assume that the languages are mentioned in the following order:
\begin{itemize}
\item highest priority --- non-trivial pattern languages and infinite unions of pattern languages. If both specializations wrt $X_i$ and $X_j$ yield such languages for $X_k$-projection, then we mention them using disjunction.
\item average priority --- trivial pattern languages. If $X_i$-specialization yields a trivial pattern language, and $X_j$-specialization yields a thin language, then we mention only the former.
\item low priority --- thin languages. An $X_k$-projection cell is marked as ``thin'' iff any specialization wrt any variable not equal to $X_k$ yields a thin $X_k$-projection language. 
\end{itemize}

\begin{longtable}{|rlccc|}
\hline & {\textbf{Equation}} & {$X_1$-\textbf{proj}} 
& {$X_2$-\textbf{proj}}
& {$X_3$-\textbf{proj}}\\ \hline 
\endfirsthead

\multicolumn{5}{c}%
{{\bfseries \tablename\ \thetable{} -- continued from previous page}} \\
\hline & {\textbf{Equation}} & {$X_1$-\textbf{proj}} 
& {$X_2$-\textbf{proj}}
& {$X_3$-\textbf{proj}}\\ \hline 
\endhead

\hline \multicolumn{5}{|r|}{{Continued on next page}} \\ \hline
\endfoot

\hline \hline
\endlastfoot

1.& $X_1 X_2\weql X_2 X_1$ & thin & thin & -\\
2.& $X_1^2 X_2^2\weql X_3^2$& thin & thin & thin \\
3.& $X_1 X_3\weql X_2 X_1$&thin&$[\omega X]$&$[X\omega]$\\
4.& $X_1 X_2 X_3\weql X_3 X_1 X_2$&thin & $[X]$& $[\omega X]$ or $[X\omega]$\\
5.& $X_1 X_2 X_3\weql X_3 X_2 X_1$& $[X\omega]$ and $[\omega X]$ & $[X]$& $[X \omega]$ and $[\omega X]$\\
6.& $X_1 X_2 X_3^2\weql X_3^2 X_2 X_1$&$[\omega^2 X]$ and $[X\omega^2]$ or inf & $[X]$ & $[\omega X]$ and $[X\omega]$\\
7.& $X_1 X_2 X_3\weql X_2 X_3 X_4$ &\multicolumn{2}{l}{(4-variable equation, omitted)}&\\
8.& $X_1 X_2 X_3 X_3\weql X_3 X_2 X_3 X_1$& $[\omega X]$ and $[X\omega]$ & inf & $[\omega X]$ and $[X\omega]$\\
9.& $X_1 X_2 X_2 X_3\weql X_2 X_3 X_1 X_2$&$[\omega X]$& thin & $[X\omega]$\\
10.& $X_1 X_2 X_1 X_3 X_2\weql X_3 X_2 X_1 X_2 X_1$& $[X\omega]$ & thin & inf\\
11.& $X_1 X_3 X_3 X_2^2\weql X_3 X_2^2 X_1 X_3$& inf &$[X\omega]$ & thin\\
12.& $X_1 X_2 X_3 X_2 \weql X_2 X_3 X_2 X_1$& inf & thin & inf \\
13.& $X_1 X_2 X_3^2 \weql X_2 X_3^2 X_1$& $[\omega X]$ or $[X\omega]$ & $[X]$ & $[X]$\\
14.& $X_1 X_3 X_2 X_3 \weql X_2 X_3 X_3 X_1$&$[\omega X]$ or $[X\omega]$ &inf & $[X]$\\
15.& $X_2 X_1 X_3 X_3 X_1^2 \weql X_3 X_1^2 X_2 X_1 X_3$& thin & thin & thin \\
16.& $X_1^\alpha \weql X_2^\beta$& thin & thin & -\\
17.& $X_1 X_2 X_3 \weql X_2^\alpha X_1$& thin& thin &thin\\
18.& $X_1 X_2^{\alpha+1} X_3 \weql X_2^\beta$&thin & thin & thin\\
19.& $X_1 X_3 X_1 \weql (X_2 X_3)^{\alpha+2}$&thin &thin & thin\\
20.& $X_3 X_1^2 \weql (X_2 X_3)^{\alpha+2}$& thin &thin & thin\\
21.& $X_1^{\alpha+2} \weql (X_2 X_3)^{\beta+2} X_2$&thin &thin &thin\\
22.& $X_1 X_3 \weql X_2^{\alpha} X_1$&thin& $[\omega X]$ & $[X\omega]$\\
23.& $X_1 X_3^{\alpha+1} \weql X_3^{\alpha+1} X_2$&$[\omega X]$& $[X\omega]$& thin\\
24.& $X_1^{\alpha+2} \weql X_2 X_3 X_2$&$[\omega X]$ and $[X\omega]$ &thin& inf\\
25.& $X_2 (X_3 X_2)^{\alpha+1} X_1 \weql X_1 X_2 (X_3 X_2)^{\alpha+1}$&$[\omega X]$ and $[X\omega]$ or inf & $[X]$ & inf\\
26.& $X_1 X_2 X_3^{\alpha+2} \weql X_2 {X_3}^{\alpha+2} X_1$&$[X\omega]$ &$[X]$& thin\\
27.& $X_1 X_3^{\alpha+2} X_2 X_3 \weql X_2 X_3 X_1 {X_2}^{\alpha+2}$&$[\omega X]$ & thin & $[X\omega]$\\
28.& $X_1 X_3^{\alpha+2} X_2 \weql X_2 {X_3}^{\alpha+2} X_1$& inf & inf & -

\label{Table::3eq}
\end{longtable}

The equation 28 has no strongly primitive $X_1$- and $X_2$-solutions, hence, its specialization wrt the given variables is not possible.

The equations 29--33 are omitted, hence they contain more than 3 variables.

\begin{lemma}\label{Lemma::NFA}
Given distinct strings $\omega_1$, ..., $\omega_n$ being not substrings of each other, a minimal non-deterministic automaton recognizing a language $L$ of strings containing all the substrings $\omega_1$,..., $\omega_n$, in a large enough alphabet, contains at least $2^n\times \min_{i\in\{1,n\}}(|\omega_i|)$ states.
\begin{proof}
Let $\#$ be a letter not contained in $\omega_1...\omega_n$.  Consider the equivalence classes determined by all possible subsets of $\{1,...,n\}$ in a following way. Given $M\in 2^{\{1,...,n\}}$, word $\omega_M$ is concatenation of the substrings $\bigl\{\omega_i\#\mid i\in M\bigr\}$ in the increasing order wrt index $i$. Then, for any $\omega_{M_1}$, $\omega_{M_2}$, $M_1\not\subset M_2$, the word $\omega_{\{1,...,n\}\setminus M_1}$ discerns the classes $\omega_{M_1}$ and $\omega_{M_2}$, moreover, $\omega_{M_2}\omega_{\{1,...,n\}\setminus M_1}\notin L$. Hence, the classes must correspond to distinct NFA states~\cite{Birget}.

The upper-triangular matrix verifying the lower bound of the number of NFA states is given in Figure~\ref{eqclasses}. Rows are marked by string prefixes, columns are marked by suffixes, and the cell on $i$-th row and $j$-th column contains $1$ iff the concatenation of the corresponding prefix and suffix belongs to the language $L$. 
\end{proof}
\end{lemma}

\begin{figure}[H]

\centering\small
$
\begin{array}{rcccccccc}
& \varepsilon & \omega_1\# & \omega_2\# & \dots & \omega_1\#\omega_2\# & \dots & \omega_{i_1}\#\omega_{i_2}\dots \#\omega_{i_k}\# & \dots\\
\omega_1\#\omega_2\#\dots \#\omega_n\# & 1 & 1 & 1 & \dots & 1 & \dots & 1 & \dots \\
\omega_2\#\omega_3\#\dots \#\omega_n\# & 0 & 1 & 0 & \dots & 1 & \dots & 1\in \{i_1, ...,i_k\} & \dots \\
\omega_1\#\omega_3\#\dots \#\omega_n\# & 0 & 0 & 1 & \dots & 1 & \dots & 2\in \{i_1, ...,i_k\} & \dots \\
\omega_1\#\dots \omega_{j-1}\#\omega_{j+1}\# & \multirow{2}{*}{0} & \multirow{2}{*}{0} & \multirow{2}{*}{0} & \multirow{2}{*}{\dots} & \multirow{2}{*}{0} & \multirow{2}{*}{\dots} & \multirow{2}{*}{$j\in \{i_1, ...,i_k\}$} & \multirow{2}{*}{\dots} \\
\dots \#\omega_n\# &  &  &  &  &  &  &  &  \\
\omega_3\#\dots\#\omega_n\# & 0 & 0 & 0 & \dots & 1 & \dots & \{1,2\}\subseteq \{i_1, ...,i_k\} & \dots \\
\dots & 0 & 0 & 0 & \dots & 0 & \dots & \dots & \dots \\
\multirow{2}{*}{$\omega_{j_1}\#\dots \#\omega_{j_l}\#$} & \multirow{2}{*}{0} & \multirow{2}{*}{0} & \multirow{2}{*}{0} & \multirow{2}{*}{\dots} & \multirow{2}{*}{0} & \multirow{2}{*}{\dots} & \{1,...,n\}\setminus\{j_1,\dots,j_l\} & \multirow{2}{*}{\dots} \\
 &  &  &  &  &  &  & \qquad\qquad\subseteq \{i_1, ...,i_k\} &  \\

\end{array}
$
\caption{The upper-triangular matrix verifying the lower bound on the NFA states space. The rows correspond to prefixes $u_i$, columns correspond to suffixes $v_j$, a boolean value in the cell $(i,j)$ shows whether the word $u_i v_j$ is in the given language.}

\label{eqclasses}
\end{figure}

\subsection{Proofs of Reduction Lemmas}

\subsubsection{Proof of Lemma~\ref{Lemma:standaloneReduction}}

We recall that $\infopat{\absval}=\begin{cases}\wequat{Z}{\upsilon_1Y_0}\\ \bigcap^n_{i=1}\wequat{Z}{X_i\omega_i Y_i}\\ \wequat{Z}{X_0 \upsilon_2}\end{cases}$ and is basically reduced, $S = \Sol_Z\bigl(\infopat{\absval}\bigr)$. 

\lemsinglereduction*
\begin{proof}
Let $\Sigma$ contain at least two letters, say $a$ and $b$. First, assume that the unavoidable in $S$ word violating the reduced-form condition above is of the form $\delta\Phi \delta$, where $\delta\in\Sigma$ is arbitrary, $\Phi\in\Sigma^*$. Without loss of generality, we assume $\delta=a$. Let $\tau=b^p$, where $p=\sum_{i=1}^n |\omega_i|+|\nu_1|+|\nu_2|+1$. Note that $\tau$ cannot be a substring of any unavoidable word. This fact allows us to use $\tau$ as a delimiter, since, for any $\omega_i\tau \omega_{j}$ including the unavoidable word, if the unavoidable word contains a letter positioned in $\omega_j$, it cannot contain any letter positioned in $\omega_i$ by the choice of $\tau$.

Construct the following word:

$$\nu_1\tau\omega_1\tau\dots\tau\omega_n\tau\nu_2$$

Since $a\Phi a$ is neither a substring of $\nu_1$, $\nu_2$, nor of any $\omega_i$, $a\Phi a$ must include $\tau$, but $\tau$ is not unavoidable in $S$. Contradiction.

Hence, any unavoidable in $S$ word not being a subword of a word from $\infopat{\absval}$ must start and end with distinct letters. Say, let such an unavoidable word be $a\Phi b$, where $\Phi\in\Sigma^*$.

If $|\Sigma|>2$, we choose the delimiter $\tau=c^p$, where $c\neq a$ and $c\neq b$, and use the reasoning above to show that $a\Phi b$ cannot be unavoidable.

If $|\Sigma|=2$, consider the delimiter $\tau=a^p$. For uniformity, let $\omega_0=\nu_1$, $\omega_{n+1}=\nu_2$ The following word $\Gamma_0$:

$$\omega_0\tau\omega_1\tau\dots\tau\omega_n \tau \omega_{n+1}$$

includes $a\Phi b$, because $a\Phi b$ is assumed to be unavoidable, hence $$\exists i_1,\dots, i_k, s_1,\dots s_k,\xi_1,\dots, \xi_k\forall 1 \leq j \leq k(a^{s_j}\omega_{i_j}=a\Phi b \xi_j).$$ For the set of such words $\omega_{i_j}$ we use another delimiter $\tau_1 = a^p b$, preserving the delimiter $a^p$ for the rest. Additionally, we rearrange the subwords of $\Gamma_0$ in such a way that all the words $\omega_{i_j}$ are grouped at its suffix. Given the ending subword $\omega_{n+1}$, if $n+1\notin \{i_1,\dots, i_k\}$, let $\tau_2=\tau$, otherwise let $\tau_2 = \tau_1$. So, we construct the following word $\Gamma$ and try to identify position of the unavoidable $a\Phi b$ in it.

$$\underbrace{\omega_0\tau\omega_{t_1}\tau\dots\tau\omega_{t_l}}_{\omega_{t_q}\notin \{\omega_{i_1},\dots,\omega{i_k}\}}\overbrace{\tau_1\omega_{i_1} \tau_1 \dots \tau_1 \omega_{i_k} \tau_2 \omega_{n+1}}^{\text{must include }a\Phi b}$$

The subword $a\Phi b$ cannot occur in the prefix containing $\omega_{t_q}$ subwords, by the choice of $\omega_{t_q}$. Hence, there are some words $\omega_{i_{r_1}}$, \dots, $\omega_{i_{r_m}}$ s.t. $a^{k_2} b \omega_{i_{r_j}}$ starts with $a\Phi b$. Consider any such $\omega_{i_{r_j}}$. By its choice, the following conditions hold:

$$\begin{cases}
a^{k_1} \omega_{i_{r_j}} = a\Phi b\xi_1\\
a^{k_2} b \omega_{i_{r_j}} = a\Phi b\xi_2
\end{cases}$$

Therefore, $k_2 = k_0 + k_1$. Let $\Phi = a^{k_0+k_1-1} b \Phi' b$ (hence, we do not consider the case when $\Phi\in a^+ b$). Then

$$\begin{cases}
\omega_{i_{r_j}} = a^{k_0} b \Phi' b\xi_1\\
\omega_{i_{r_j}} = \Phi' b\xi_2
\end{cases}$$

Hence $\xi_2 = a^{k_0} b \xi_1$, and the equation $\Phi' b a^{k_0} b = a^{k_0} b \Phi' b$ holds. Then, either $k_0=0$ and $\Phi'\in b^*$ (hence, $\Phi\in a^+ b^+$) or $\Phi'\in (a^{k_0} b)^* a^{k_0}$. 

In the latter case, we can replace in $\Gamma$ all the $\tau_1$ occurrences by $a^p b^2$ in order to avoid $a\Phi b$. 

If $\Phi=a^{k_1+1}b^{k_2+1}$, then we modify $\Gamma$ as follows. If $\omega_{i_{r_j}}$ ends with $a$, replace $\tau_1$ occurrence next to it by $\tau'_{i_{r_j}} = ba^{k_1+1}b^{k_2+1}$, otherwise replace it with $\tau'_{i_{r_j}} =a^{k_1+1}b^{k_2+1}$. Hence, $a^{k_1+2} b^{k_2+2}$ cannot occur in $\omega_{i_{r_j}}\tau'_{i_{r_j}}$.

Hence, the only possible cases for $a^{k_1} b^{k_2}$ to be unavoidable are the cases when $k_1=1$ or $k_2=1$, which concludes the proof.
\end{proof}

Now the unavoidable words that must occur in the set of words containing subwords from $\infopat{\absval}$ can be easily constructed. Let $\Sigma = \{a,b\}$. Internal strings from $\infopat{\absval}$ are the strings determining the equations $\wequat{Z}{X_i \omega_i Y_i}$.

\begin{itemize}
\item If $\wequat{Z}{\omega a^k Y}\in\infopat{\absval}$, and there is at least one another equation in $\infopat{\absval}$ containing $b$, then $a^k b$ is unavoidable w.r.t. $\infopat{\absval}$. 
\item Given two internal strings $\omega_1 b a^{k_1}$ and $ \omega_2 b a^{k_2}$ in $\infopat{\absval}$ equations, if $\omega_1$ is not a suffix of $\omega_2$ and vice versa, $a^{\min(k_1, k_2)} b$ is unavoidable.
\item If $\wequat{Z}{X a^k \omega}\in\infopat{\absval}$, and there is at least one another string in $\infopat{\absval}$ equations containing $b$, then $b a^k$ is unavoidable w.r.t. $\infopat{\absval}$. 
\item Given two internal strings $a^{k_1} b \omega_1$ and $a^{k_2}b \omega_2$ in ${\infopat{\absval}}$, if $\omega_1$ is not a prefix of $\omega_2$ and vice versa, $ba^{\min(k_1, k_2)}$ is unavoidable.
\item Symmetrically, the unavoidable words $ab^{\min(k_1, k_2)}$ and $b^{\min(k_1, k_2)} a$ can be found.
\end{itemize}

\begin{algorithm}[htb]
\caption{Algorithm for finding unavoidable words of the form $a^k b$ with respect to $\infopat{\absval}\in\FactorCode$.}

There $\mathcal{U}$ consists of pairs $\langle k, \omega_i\rangle$, where the first letter of $\omega_i$ is not $a$ and  $a^k \omega_i\in\infopat{\absval}$ is an internal substring (determining the equation $\wequat{Z}{Xa^k \omega_i Y}$). The list $\mathcal{U}$ is sorted by $k$ value decreasing.
\begin{algorithmic}[1]
\STATE $\mathcal{U} \gets \text{sortByAkPrefixes}(\infopat{\absval})$
\STATE $i \gets 1$ 
\STATE $\langle k_{\max}, \omega_{\max}\rangle \gets \mathcal{U}[i]$
\WHILE{$k_{max} > 1$ \AND $i \leq |\mathcal{U}|$}
    \STATE $(k_1, \omega_{\text{next}}) \gets \mathcal{U}[i+1]$
\STATE
\LONGCOMMENT{If two words in a factor code share a common maximal prefix $a^k$, then none of them is a prefix of another}
    \IF{$k_1 == k_{\max}$}
        \STATE \textbf{return} $k_{\max}$
    \ENDIF
    \IF{\NOT($\omega_{\max}.\text{isPrefixOf}(\omega_{\text{next}})$)}
        \STATE \textbf{return} $k_{\max}$
    \ELSE
        \STATE $k_{\max} \gets k_1$
        \STATE $\omega_{\max} \gets \omega_{\text{next}}$
        \STATE $i \gets i + 1$
    \ENDIF
\ENDWHILE
\STATE \textbf{return} $k_{\max}$
\end{algorithmic}

\end{algorithm}

\subsubsection{Proof of Lemma~\ref{Lemma:ReducedLength}}

\lemlenreduction*
\begin{proof}
First, we can easily construct a concrete string value satisfying all the equations given in $\objval{\absval}$, and having any length equal and more than $\displaystyle\sum_{i=1}^k |\omega_i|+|\upsilon_0|+|\upsilon_1|$. 

Second, let us assume that there exists an abstract object $\absval_*$ with the same concretisation set as $\absval$, but with $\objval{\absval_*}\neq \objval{\absval}$. First of all, the equations restricting prefixes and suffixes of strings in their concretisation sets must coincide. Really, let us assume $\wequat{Z}{\upsilon'_0 Y}\in \objval{{\absval}_*}$, and $\upsilon'_0\neq\upsilon_0$. Let $|\upsilon'_0|\leq |\upsilon_0|$, and $\delta$ be a letter not occurring in any of $\omega_i$, $\upsilon_i$. Then the concretisation set of $\absval_*$ contains a string prefixed with $\upsilon'_0\delta$, while the concretisation set of $\absval$ cannot contain such a string. The same reasoning proves that the equation in ${\absval}_*$ determining the left ideal is $\wequat{Z}{X\upsilon_1}$.

Now we reason by recursion on $k$. From the set $\{\omega_1,\dots, \omega_k\}$, we choose the smallest $\omega_{i_1}$ wrt the length-lexicographic order. As before, $\delta$ is a letter not occurring in any of $\omega_i$ and $\omega'_i$.
\begin{itemize}
\item If there is some $\omega'_{i_1}$ s.t. $|\omega_{i_1}|>|\omega'_{i_1}|$, then $\concr{\absval_*}$ contains a string prefixed with $\upsilon_0 \delta\omega'_{i_1}\delta$, while $\concr{\absval}$ does not contain such a string.
\item If there is no $\omega'_{i_1}$ s.t. $\omega_{i_1}$ is its prefix, or any $\omega'_{i_1}$ starting with $\omega_{i_1}$ is longer than $\omega_{i_1}$, then $\concr{\absval}$ contains a string prefixed with $\upsilon_0 \delta \omega_{i_1}\delta$, while $\concr{\absval_*}$ does not contain such a string.
\end{itemize} 
Hence, the only option in which $\concr{\absval}$ and $\concr{\absval_*}$ can coincide is the case when $\concr{\absval_*}$    contains $\omega_{i_1}$.

Recursively continuing this reasoning, we prove that $\objval{\absval}$ and $\objval{\absval_*}$ contain the same set of equations.
\end{proof}

\subsubsection{Proof of Lemma~\ref{Lemma:Propagation}}

\lempropagation*
\begin{proof}
Let us assume the contrary: let the reduced representation of $\infopat{\absval}_1$ contain an equation $\wequat{Z}{X\omega Y}$ such that it is not contained in $\infopat{\absval}_1$ and $\sigma_2$ does not preserve $\omega$. The cases of equations $\wequat{Z}{\omega Y}$ and $\wequat{Z}{X\omega}$ are considered similarly.

By the assumption, $\omega$ is contained as a subword in all the words in $\sigma_1(\concr{\sigma^{-1}_1(\infopat{\absval}_1)})$. Now we show how to construct a word s.t. it satisfies all the constraints of $\infopat{\absval}_1$, but does not contain the subword $\omega$.

Take one-letter word $\delta$ neither starting nor ending $\omega$. Such a word exists, because $\sigma_1$ and $\sigma_2$ are non-erasing, and $\sigma_2\succ\sigma_1$. Given $\upsilon_0$ as a prefix, $\upsilon_1$ as a suffix, and $\omega_i$ as infixes, construct a skeleton of the counterexample with the parameter $\tau$:

$$\overbrace{\upsilon_0\delta^{|\omega|} \omega_1 \delta^{|\omega|} \dots \omega_k \delta^{|\omega|}}^{\text{does not contain }\omega} \tau \delta^{|\omega|} \upsilon_1$$

By construction, $\omega$ can occur only in the $\tau$ part there, crossing no $\delta^{|\omega|}$ bound.

Given an equation $\wequat{Z}{X_i \tau_i Y_i}$ defined by $\sigma_2$, specify the corresponding infix of $\tau$. If $|\tau_i|<|\omega|$, the infix is $\delta \tau'_i\delta$, where $\tau'_i$ is an arbitrary element of the inverse image of $\tau_i$ wrt $\sigma$, i.e. $\{\upsilon\mid \sigma(\upsilon)=\tau_i\}$. If $|\tau_i|\geq |\omega|$, choose an element from its inverse image as follows. The element is accumulated in the $\upsilon_i$, initially set to $\empt$. The suffix of $\tau_i$, $\tau_{k,i}$, is initially set to $\tau_i$, while $k$ is set to zero.

Let $\tau'_{k,i}$ be a maximal prefix of $\tau_{k,i}$ preserved by $\sigma_2$. Choose its maximal suffix $\hat{\tau}_{k,i}$ s.t. $\sigma_2^{-1}(\hat{\tau}_{k,i})$ (which is denoted by $\tau''_{k,i}$ below)  starts $\omega$. 

\begin{itemize}
\item If $\tau'_{k,i}=\tau_{k,i}$ (i.e. the remaining part of $\tau_i$ is preserved by $\sigma_2$), let $\upsilon_i \mapsto \upsilon_i\tau_{k,i}$ and end the loop.
\item Otherwise, since $\sigma_2$ does not preserve $\omega$, by its choice, $\omega=\tau''_{k,i}\delta_1\omega_1$, where $\delta_1$ is not preserved by $\sigma_2$. Hence, for every letter in $\nu_i$ immediately following $\tau'_{k,i}$, $\delta'_1$, $\sigma_2^{-1}(\delta'_1)$ contains an element $\hat{\delta}_{1}$  not equal to $\delta_1$. Let $\upsilon_i\mapsto \upsilon_i\sigma_2^{-1}(\tau'_{k,i})\hat{\delta}_{1}$. Set $\tau_{k+1,i}$ to be the remaining suffix of $\tau_{k,i}$ without $\tau'_{k,i}\sigma_2(\delta_1)$, and increment $k$.
\end{itemize}

Continuing this procedure, we obtain the counterexample --- a string that satisfies the given property but does not contain $\omega$.
\end{proof}

\subsection{Some restrictions in word-equation-based approach}

\begin{lemma}\label{lemma:expressibility}
{There are regular languages having morphic images whose inverse is the given language are not representable by languages of word equations.}
\end{lemma}

\begin{proof}
Let us consider $\Lang_0=(ab|ba)^*$. Let $\sigma(a)=\omega_1$, $\sigma(b)=\omega_2$, then $\sigma((ab|ba)^*)=(\omega_1\omega_2|\omega_2\omega_1)^*=\Lang_1$. We assume that an inverse $\sigma^{-1}(L_1)$ is a maximal set of words in alphabet $\{a,b\}$ s.t. $\forall \upsilon\in \sigma^{-1}(\Lang_1)(\sigma(\upsilon)\in \Lang_1)$, and that $\sigma^{-1}(\Lang_1)=\Lang_0$.

The last condition imposes an obvious restriction on $\sigma$: $|\omega_1|>0$ and $|\omega_2|>0$, and the alphabet of $\omega_1\omega_2$ is not unary.

Now we refer to a following lemma of~\cite{Day}:

If a thin regular language $\Lang^*$ (i.e. a Kleene star of a regular $\Lang$) is represented by a word equation, then for all $\tau_1,\tau_2\in \Lang$, $\tau_1 \tau_2 = \tau_2 \tau_1$.

We recall that a language $\Lang$ in alphabet $\Sigma$ is said to be thin iff there exists at least one word $\omega\in\Sigma^+$ that is avoided as a subword in elements of $\Lang$. That is, $\forall u\in \Lang(u\neq u_1\omega u_2)$.

Let us show that $\Lang_1$ is thin. If $|\omega_1|=|\omega_2|=1$, the fact trivially is implied from the fact that $\Lang_0$ is thin (e.g., words in $\Lang_0$ never include $aaa$). Let $|\omega_1|+|\omega_2|\geq 3$. We count a number of possible substrings in $\Lang_1$ of the length $|\omega_1\omega_2\omega_1\omega_2|$. Such a substring may include either:

\begin{itemize}
\item two occurrences of $\omega_1\omega_2$ and $\omega_2\omega_1$, or their single occurrences combined in any order, giving a total of 4 variants;
\item a single occurrence of $\omega_1\omega_2$ or $\omega_2\omega_1$, prefixed and suffixed by two other occurrences, a total of $|\omega_1+\omega_2|\cdot 8$ variants.
\end{itemize}

Hence, the number of substrings of the length $|\omega_1\omega_2\omega_1\omega_2|$ in words of $\Lang_1$ is at most $4+|\omega_1+\omega_2|\cdot 8$. While a total number of the substrings is $2^{|\omega_1\omega_2\omega_1\omega_2|}$, which is greater than the given upper bound for any $|\omega_1|+|\omega_2|\geq 3$.

We have shown that $\Lang_1$ is thin. Since $\Lang_1=(\omega_1\omega_2|\omega_2\omega_1)^*$, by Day et al, 
$$\exists \tau,n\biggerl\tau\text{ is primitive }\&\omega_1\omega_2=\tau^n\&\omega_2\omega_1=\tau^n\biggerr.$$ 
Therefore, both $\omega_1$ and $\omega_2$ are powers of $\tau$, say, $\omega_1=\tau^{k_1}$, $\omega_2=\tau^{k_2}$. Then $\sigma(a^{k_1+k_2})=\tau^{k_1\cdot(k_1+k_2)}=(\omega_1\omega_2)^{k_1}\in \Lang_1$, while $a^{k_1+k_2}\not\in \Lang_0$.
\end{proof}

\end{document}